\documentclass[aoas]{imsart}
\usepackage[linesnumbered,ruled,vlined]{algorithm2e}
\usepackage{upgreek}
\usepackage{threeparttable}
\usepackage{multirow}
\RequirePackage[OT1]{fontenc}
\RequirePackage{amsthm,amsmath}
\RequirePackage{natbib}  
\RequirePackage[colorlinks,citecolor=blue,urlcolor=blue]{hyperref}
\RequirePackage{graphicx}
\RequirePackage{dsfont}

\RequirePackage{bm}  

\startlocaldefs
\theoremstyle{plain}

\newtheorem{theorem}{Theorem}[section]
\newtheorem{lemma}[theorem]{Lemma}
\theoremstyle{remark}

\endlocaldefs

\begin{document}
	
	\begin{frontmatter}
		
		\title{An Efficient Doubly-robust Imputation Framework for Longitudinal Dropout, with an Application to an Alzheimer's Clinical Trial}
		
		\begin{aug}
			\author[A]{\fnms{Yuqi} \snm{Qiu}\ead[label=e1]{yqqiu@fem.ecnu.edu.cn}}
			\and
			\author[B]{\fnms{Karen} \snm{Messer}\ead[label=e2]{kmesser@health.ucsd.edu}}
			\address[A]{School of Statistics,
				East China Normal University,
				\printead{e1}}
			
			\address[B]{Division of Biostatistics and Bioinformatics,
				University of California San Diego,
				\printead{e2}}
		\end{aug}

		\begin{abstract}
			
			We develop a novel doubly-robust (DR) imputation framework for longitudinal studies with monotone dropout, motivated by the informative dropout that is common in FDA-regulated trials for Alzheimer's disease.  In this approach, the missing data are first imputed using a doubly-robust augmented inverse probability weighting (AIPW) estimator,  then the imputed completed data are substituted into a full-data estimating equation, and the estimate is obtained using standard software. 
			The imputed completed data may be inspected and compared to the observed data, and standard model diagnostics are available. The same imputed completed data can be used for several different estimands, such as subgroup analyses in a clinical trial, allowing for reduced computation and increased consistency across analyses.	We present two specific DR imputation estimators, AIPW-I and AIPW-S, study their theoretical properties, and investigate their performance by simulation.  AIPW-S has substantially reduced computational burden compared to many other DR estimators, at the cost of some loss of efficiency and the requirement of stronger assumptions. Simulation studies support the theoretical properties and good performance of the DR imputation framework.  Importantly, we demonstrate their ability to address time-varying covariates, such as a time by treatment interaction.  We illustrate using data from a large randomized Phase III trial investigating the effect of donepezil in Alzheimer's disease,  from the Alzheimer's Disease Cooperative Study (ADCS) group.

		\end{abstract}
		
		
	\end{frontmatter}

	\section{Background and a motivating example}
	Dropout rates of 25\% or more are common in large randomized trials for Alzheimer's disease and other dementias.  The dropout rate is generally higher for patients with more severe disease and also for patients receiving active treatment ({\bf Figure} \ref{MCIplot}). Thus, estimates of treatment effect from the trial may  be substantially biased unless the dropout is properly accounted for.  
	However, the primary analysis of such a trial is tightly pre-specified.  US Food and Drug Administration (FDA) guidance and practices support the use of a  mixed-effects  model with repeated measures (MMRM) or 
	a generalized estimating equations (GEE) approach with a restricted set of  covariates;  the primary analysis usually tests a model-adjusted estimate of a contrast between treatment arms.   Because of restricted covariates and modeling strategies, it is very possible that the covariate adjustment may be insufficient to assure unbiased estimates from the model.  To address this problem in the regulatory setting of clinical trials, additional imputation-based sensitivity analyses are recommended to assess  the potential  bias from dropout, usually through sequential multiple imputation by chained equations \citep{FDAguidance, EMAguidance}.  
	
	Motivated by this setting, we investigate a principled doubly-robust approach to such imputation for longitudinal dropout, derived from the theory of optimal augmented inverse probability weighted (AIPW) estimators, as developed by \cite{Bang2005,Tsiatis2006,  Seaman2009, Tsiatis2011,Rot2012, Schnitzer2016} and others in the longitudinal setting.    Formally, the interest lies in obtaining a robust, consistent and asymptotically normally distributed (CAN) estimate of an estimand  $\beta$, defined as the solution to a pre-specified   estimating equation $E[ U(\beta)]=0$.  
	A simple example is where $\beta = E(Y)$ and $U(\beta) = Y - \beta$.
	If $U(\beta)$ is applied to the "full" data (fully observed data with no dropout), the solution  $\hat \beta$ is a CAN estimator of $\beta$.  However, the primary intent-to-treat analysis from the trial usually applies $U(\beta)$ to the observed data. Under a missing at random (MAR) assumption conditional on covariates, this would indeed provide a consistent estimator  if all  covariates are appropriately controlled for.  As explained above, however, in our setting $U$ is usually a GEE  derived from a generalized linear model with limited covariates, rather than a fully specified likelihood. Therefore, the approach in practice is generally not assumed to provide unbiased estimates in the presence of missing data.   
	
	Consistent estimators in  the  setting of longitudinal monotone dropout have been well studied in theory, including doubly-robust (DR) estimators; for a recent review see 
	\cite{Seaman2018}.   \cite{ChoPaik1997} constructed an imputation approach to obtain consistent estimators for longitudinal data with dropout, using a sequential regression algorithm to impute the missing outcome values.   \cite{Robins1995} 
	noted that  consistent estimates of $\beta$ may  be obtained by solving an  inverse probability weighted (IPW) version of the  estimating functions $U(\beta)$ applied to the data from completers only,  or by using an AIPW estimator, which increases the efficiency by augmenting the IPW estimator with the observed information from dropout subjects using a regression model. 	%
	\cite{Scharfstein1999} noticed that these AIPW estimators are doubly-robust (DR), in that they are consistent when either the regression model or the IPW model is correctly specified. 
	\cite{Bang2005} proposed a   regression representation of an AIPW estimator, using a recursive regression approach which incorporates the IPW's as covariates; however in the longitudinal setting they only explicitly studied the case where $\beta =E [Y]$ at last visit.  \cite{Tsiatis2006} developed the theory of  optimal AIPW estimating equations in the  setting of longitudinal data with monotone dropout for a general class of $U(\beta)$, and implemented an improved DR estimator with MMRM as the regression model in \cite{Tsiatis2011}, which, however, is quite complex in practice. 
	Following \cite{Tsiatis2006},  \cite{Seaman2009} considered $U(\beta)$ arising as generalized estimating equations, and used Paik's sequential regression framework to construct an AIPW estimator for the regression coefficients.   \cite{Rot2012}  further developed an optimal AIPW estimator which may obtain good efficiency when the outcome model is misspecified, and provided simulation studies in the cross-sectional setting. 
	More recently, 
	\cite{Schnitzer2016} adapted the \cite{Bang2005} approach in the longitudinal setting to include more general estimands defined by a GLM with baseline covariates, however the class of estimands in particular does not include  the regression coefficient of a time-varying covariate. \cite{Schnitzer2016} also study a closely related targeted maximum likelihood estimator (TMLE) \citep{van2006} 
	and show that it has similar performance as the adapted Bang and Robins method by simulation.  \cite{Long2012} and \cite{Hsu2016} considered incorporating a DR estimator into a multiple imputation (MI) approach in the cross-sectional setting, however this differs from the longitudinal  approach we consider here.

	It is worth noting that, compared with the cross-sectional setting,  DR methods for longitudinal data are comparatively less well-developed and less often used in practice.  In particular, we know of relatively few simulation studies  that study these longitudinal doubly-robust estimators \citep{Seaman2009, Tsiatis2011, Schnitzer2016}.  While several approaches can in theory estimate the coefficient of a  time varying covariate  \citep{Seaman2009, Tsiatis2011}, we did not find simulation studies that cover this case which is critical to the clinical trials setting. We are unfamiliar with examples of the use of DR estimators in clinical trials, where we think they may have a useful and important role.
	
	In this paper, we further develop results from  \cite{Seaman2009}  into a general  imputation-based framework for construction of a locally efficient doubly-robust estimator in the setting of longitudinal data with monotone dropout.  
	The approach is to impute a complete dataset using a suitable doubly-robust estimator, and then apply the estimating functions $U(\beta)$ to the fully-imputed data.  We first investigate a standard AIPW-based doubly-robust estimator (AIPW-I) within the framework, and then we propose a simpler doubly-robust estimator (AIPW-S) also within the framework. Confidence intervals are provided by the bootstrap, in an imputation setting which is familiar to practitioners of clinical trials. This imputation framework has several advantages compared to existing approaches in the literature: 1) once the imputed completed data are obtained, the same completed dataset will support doubly-robust estimation of additional estimands, such as additional additional group or subgroup effects, providing computational efficiency and consistency across multiple analyses;  2) the imputed completed data can be inspected using standard descriptive statistics, to better understand the behavior of the resulting doubly-robust estimator and any differences from the primary analysis  (which uses the observed data directly).   Both the imputation step and the analysis step use well-understood modeling approaches which are easy to pre-specify and are suitable for the clinical trials setting. The imputation framework is applicable to any AIPW-based doubly-robust estimator, including the existing highly developed doubly-robust methods and we hope will facilitate the construction of doubly robust estimates of causal effects in clinical trial settings.  Importantly, our proposed framework supports doubly robust estimation of the coefficient of a  time-varying covariate as the estimand, such as a treatment by time interaction term.
	
	
	The outline of the approach is as follows: for $U(\beta)$ a GEE, \cite{Seaman2009} constructed a DR estimator of regression coefficients by substituting Paik's sequential regression into the optimal longitudinal AIPW equations of \cite{Tsiatis2006}, and solving  the resulting equations by Newton-Raphson.  Our contribution is to notice that the resulting estimating equations might be rearranged into an imputation form, after establishing some algebraic identities, and to exploit the result.  The resulting estimator  only requires standard software tools, and  inherits the optimality properties of the doubly-robust imputation for a wide class of estimating functions.  Using this framework, we develop a computationally simpler doubly-robust estimator, AIPW-S, under explicit assumptions which may be practical in clinical trials and other settings.  Finally, we study the performance of these estimators through simulation, and through an application to the primary estimands of interest in a clinical trial of prodromal  Alzheimer's Disease. These simulation studies may be of independent interest, as they add to the sparse literature on longitudinal AIPW approaches.

	\subsection{Organization of this paper}
	In Section 2, we describe our motivating example. Section 3 gives notation and the details of consistent  IPW and sequential regression approaches for longitudinal data with dropout. These form the building blocks of our DR imputation method. We briefly review existing DR methods for longitudinal data in Section 4. In Section 5, we develop the AIPW-based DR imputation framework for general longitudinal data estimating equations.  In Section 6, we use the DR imputation framework to define two specific DR estimators for longitudinal data with dropout, AIPW-I and AIPW-S.  We use simulation to compare AIPW-I and AIPW-S with the original Bang and Robins estimator 
	as well as with maximum likelihood and GEE approaches, in Section 7. Section 8 presents an application to the donepezil trial in Alzheimer's disease, and Section 9 is discussion and conclusions. 
	
	\section{A motivating example: the MCI trial of donepezil}\label{MCI}
	
	\begin{figure}[h]
		\centering
		\includegraphics[width=\textwidth]{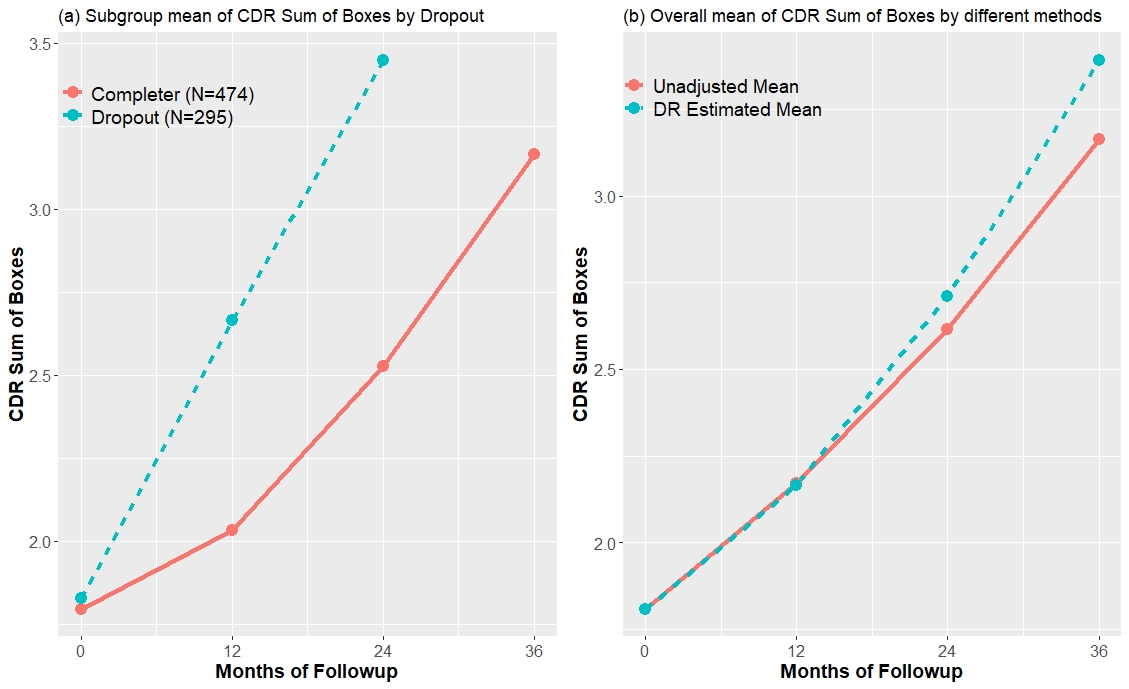}
		\caption{Time course of the mean Clinical Dementia Rating Sum of Boxes (CDR-SOB) score, using all subjects who provided data on this measure in the MCI trial (n=769). The CDR-SOB is commonly used in Alzheimer's disease to stage severity of dementia, with  a higher score indicating worse severity; it is expected to increase over time in this population with mild cognitive impairment (MCI) at baseline. Left panel:  mean CDR-SOB at each time point, for subjects who completed the study (red, 62\% of subjects) vs. subjects who dropped out before the end of the study (blue, 38\% of subjects).  Subjects who dropped out are observed to have higher change scores. Right panel:  mean CDR-SOB using the available data at each time point (red); doubly-robust estimate of the mean CDR-SOB, using all available data and adjusting for dropout (blue). }
		\label{MCIplot}
	\end{figure}
	Donepezil is a widely used cholinesterase inhibitor that improves symptoms and may delay the clinical diagnosis of Alzheimer's disease (AD) in subjects with the amnestic form of mild cognitive impairment (MCI).  A randomized, double-blind, placebo-controlled, parallel-arm trial was conducted by the Alzheimer's Disease Cooperative Study (ADCS) between March 1999 and January 2004  \citep{Petersen2005}. The study compared the time to progression from MCI to possible or probable AD  among 769 subjects with MCI who were randomized to receive donepezil (n=253),  vitamin E (n=257), or placebo (n=259) for 36 months.  Final dropout rates were 42.7\%, 38.1\%, and 32.0\% for the donepezil, Vitamin E, and placebo arms, respectively. Figure \ref{MCIplot} displays evidence of  bias due to dropout for the mean  score on the Clinical Dementia Rating sum of boxes (CDR-SOB),  one of the common primary outcome measures of AD trials.  The left panel shows that patients who eventually withdrew from the study had a much higher CDR-SOB (higher is worse) than the completers, indicating that those who dropped out had greater cognitive impairment, and that this gap increased over time. The figure on the right shows that the estimated mean from a DR approach is consistently higher than the mean of the observed data, especially at the later times with more dropout,  indicating that a DR approach may help to improve estimated effects from this trial.  These data are taken from the data archives at the ADCS;  these and similar data are often used in simulation studies that inform the design of current AD trials.

	\section{Regression modeling and inverse probability weighting approaches to longitudinal dropout}
	\subsection{Notation and data structure}
	Assume we have $N$ iid subjects potentially observed at times $j=1, \dots, M$, and for  individual $i$ at time $j$ there is data  $L_{ij}=(Y_{ij}, X_{ij})$,  where $Y_{ij}$ is a univariate outcome and $X_{ij}$ is a vector of potentially time-varying covariates; when $X_{ij}$ is independent of time, it will be simplified to a vector of always observed baseline covariates $X_{i0}$ for any $j$. Let $\bar{L}_{ij} = (L_{i1}^T,...,L_{ij}^T)^T $ denote the historical data from time $1$ to $j$. We may drop the subscript $i$ when the meaning is clear. We assume the distribution of $\bar L_{iM}$ has finite second moments.

	Each subject can potentially drop out from the study.  Let $R_{ij} \in \{0,1\}$ be a  missing indicator, so that we observe $R_{ij}$ and $(R_{ij}Y_{ij},  R_{ij}X_{ij})$ at time $j$. Under the assumption of monotone dropout, if $R_{j} = 0$ then for any $t>j$, $R_{t} = 0$.  Let $C_{ij}$ be a censoring indicator, where $C_{ij} = 1$ indicates $j$ is the final observed time for subject $i$, otherwise $C_{ij} = 0$;  let $J_i$ be the index of the last observed time point for subject $i$, so that $C_{i,J_i}=1$ for subject $i$. For completers, define $J_i = M$ and $C_{M} = R_{M} = 1$. Under the missing completely at random assumption (MCAR), $R_{j}$ is independent of $\bar L_{M}$.
	Under the missing at random assumption (MAR),  $Pr(R_{j}=0 | \bar{L}_{M}, R_{j-1}=1) = Pr(R_{j}=0 | \bar{L}_{j-1}, R_{j-1}=1)  $ so that the probability of a missing outcome depends only on previously observed data.  We also assume there is probability bounded away from zero of seeing full data over the whole support of $\bar L_{iM}$: $Pr(R_{M}=1 | \bar{L}_{M}) > \epsilon >0$, for some existing $\epsilon$.    
	\subsection{Estimating equations which define the estimand}
	We assume  that there is a 
	vector of parameters  $\mathbf{\beta}$, and a corresponding  vector of sufficiently smooth estimating functions $U(\cdot, \mathbf{\beta})$ such that $\mathbf{\beta^*}$ is the unique solution (the truth) to $E[U(\bar L_{iM}, \mathbf{\beta})]=0$.  One  of the parameters in $\mathbf{\beta^*}$ is the primary estimand of interest in the study. The solution $\mathbf{\hat \beta}$ to the full data estimating equations $\sum_{i=1}^N U_i(\bar L_{iM}, \mathbf{\beta}) =0$ is consistent for $\mathbf{\beta^*}$ and asymptotically normal, by standard arguments.    
	
	It is often assumed that the data follow a generalized linear model for the mean $\mu_{j} = E(Y_{j}|X_{j})$, with   link function  $g(\mu_{j}) = X_{j}\mathbf{\beta}$,  including of course the MMRM of the introduction. 
	Thus, $U(\bar L_{iM},{\beta})$ might be taken to be the score equations from the likelihood, or alternatively, a set of  generalized estimating equations (GEE) \citep{Liang1986} applied to the full data: 
	
	\begin{equation}\label{GEE}
		\displaystyle\sum_{i=1}^{N}U_i(\bar L_{iM})= \displaystyle\sum_{i=1}^{N} \frac{\partial \mathbf{\upmu}^T_i}{\partial \mathbf{\beta}} V^{-1}_i (\mathbf{Y}_i - \mathbf{\upmu}_i) = 0
	\end{equation}
	where $\mathbf{Y}_i = (Y_{i1},...,Y_{iM})^T$, $\mathbf{\upmu}_i = (\mu_{i1},...,\mu_{iM})^T$, and $V_i$ is an assumed working covariance matrix for $\mathbf{Y}_i$. Here, the efficient choice for $V^{-1}_i$ is the true covariance matrix of the data. However, under general regularity conditions, the solution $\mathbf{\hat \beta}$ to the full data GEE's (equation (\ref{GEE})) is consistent for $\mathbf{\beta^*}$ and asymptotically normal, for any arbitrary $V_i$.   
	
	With missing data, instead of equation (\ref{GEE}) we observe  
	\begin{equation}\label{GEEmissing}
		\displaystyle\sum_{i=1}^{N}U_i(\bar L_{iJ})=\displaystyle\sum_{i=1}^{N} \frac{\partial \mathbf{\upmu}^T_i}{\partial \mathbf{\beta}} V^{-1}_i  D_i (\mathbf{Y}_i - \mathbf{\upmu}_i) = 0
	\end{equation}
	where $D_i = \textrm{diag}(R_{i1}, \ldots, R_{iM}). $
	If the data are MCAR, the solution $\mathbf{\hat \beta}$ to equation (\ref{GEEmissing}) remains consistent, because $E[R_{j}Y_{j}]=E[Y_{j}]$, and thus equation (\ref{GEEmissing}) is a consistent estimator of $E[U_i(\bar L_{iM})]$.   However, when the dropout is MAR such that $E[R_{j}Y_{j}] \neq E[Y_{j}],$ the solution $\mathbf{\hat \beta}$ to equation (\ref{GEEmissing}) will be consistent for $\mathbf{\beta}$ if the $U$ are the score equations from the correct likelihood, but not generally otherwise. 
	\subsection{IPW estimating equations, for dropout that is MAR}\label{IPW}
	\cite{Robins1995} showed how to incorporate  inverse-probability weights into $U$ applied to observed data when the dropout is MAR. 
	Let 
	\begin{equation}\label{HZfunc}
		\lambda_{ij} =P(R_{ij}=0|R_{i,j-1}=1, \bar{L}_{i,j-1} )
	\end{equation}
	be the discrete-time hazard function of $i$ at time $j$, let 
	\begin{equation} \label{pi}
		\pi_{ij}=\displaystyle\prod_{t=1}^{j} (1-\lambda_{it}) = P(R_{ij}=1|\bar{L}_{i,j-1})  
	\end{equation}
	be the corresponding probability of being observed at time $j$, and let the weight matrix be $W_i=diag(R_{i1}/\pi_{i1},...,R_{iM}/\pi_{iM}).$  The hazard function can be consistently estimated by logistic regression if the MAR assumption holds; thus $\pi_j$ and $W_\pi$ can be consistently estimated as well. 
	Then the inverse probability weighted GEE (WGEE)  has the estimating equations: 
	\begin{equation}\label{GEE_IPW}
		\displaystyle\sum_{i=1}^{N}U^W_i(\bar L_{iJ})=\displaystyle\sum_{i=1}^{N} \frac{\partial \mathbf{\upmu}^T_i}{\partial \mathbf{\beta}} V^{-1}_i  W_i (\mathbf{Y}_i - \mathbf{\upmu}_i) = 0.
	\end{equation}

	For longitudinal data with dropout that is MAR,  WGEE provides consistent estimates of the parameters  $\mathbf{\beta^*}$ if the $\pi_{ij}$'s are consistent, because then $E[R_{j}Y_{j}/\pi_{j}] =E[Y_{j}]$, even if the working correlation is not correctly specified.
	\subsection{Regression-based sequential imputation, for dropout that is MAR}\label{Paik_algorithm}
	Alternatively, 
	\cite{ChoPaik1997} described a  sequential regression approach for imputing the missing outcome values. 
	Recall that $J_i$ is the last observed visit for subject $i$, so that $Y_{ik} $ is missing for any $k>J_i$. In order to impute $Y_{ik}$, the idea is to conduct a recursive regression process by defining $k-1$ parametric imputation models 
	$
	m^j_{k}(\bar L_j) = E[Y_{k}| \bar{L}_j, R_{j+1} = 1]
	$, where $j$ is taken from $k-1$ to $1$.  By the MAR assumption,  $E[Y_{k}| \bar{L}_j, R_{j+1}=0] = E[Y_{k}| \bar{L}_j, R_{j+1}=1]. $ Hence we may use observed data to construct a consistent estimated model $\hat m^j_{k}(\bar L_j)$, and then use the estimated model to impute missing $Y_{k}$ recursively. This process can then be repeated sequentially for $k = 2,...,M$ in order to obtain a fully imputed data set.
	
	Here we give a formal version of Paik's sequential mean imputation algorithm, which will make the relation with 
	\cite{Bang2005} more explicit :
	
	\begin{algorithm}\label{Paik}
		\KwResult{A fully imputed data $\hat Y^I_{ij}$, in which all longitudinal missing outcomes have been imputed.}
		\textbf{For $\mathbf{k = 2,...,M}$ sequentially:}\\
		\textbf{Initialize:} Identify all subjects $i$ with $J_i \geq k.$ Use these data to regress the observed values of $Y_{ik}$ on $\bar{L}_{i,k-1}$, to obtain a consistently estimated  model $\hat m^{k-1}_{k}(\bar L_{i,k-1}).$ 
		For subjects with $J_i \geq k-1$ let
		\begin{equation*}
			\hat Y^{k-1}_{ik} = 	 \left\{ 
			\begin{array}{ll}
				\hat m^{k-1}_k(\bar L_{i,k-1}) &  \mathrm{ if  } \;J=k-1    \\
				Y_{ik}  &  \mathrm{ if  }  \;J > k - 1  
			\end{array}
			\right.
		\end{equation*}\\
		\textbf{For $\mathbf{s = k-2,...,1}$ recursively:}
		Identify all subjects $i$ with $J_i \geq s+1$ and regress the values of $\hat Y^{s+1}_{ik}$ on $\bar{L}_{is}$ to obtain a consistently estimated model $\hat m^{s}_k(\bar L_{is})$ .  For all subjects with $J_i \geq s $   let
		\begin{equation*}
			\hat Y^s_{ik} = 	 \left\{ 
			\begin{array}{ll}
				\hat m^s_k(\bar L_{is}) &  \mathrm{ if  } \;J_i=s    \\
				\hat Y^{s+1}_{ik}  &  \mathrm{ if  }  \;J_i > s  
			\end{array}
			\right.
		\end{equation*}\\
		
		\textbf{Final step:}
		Output the completed data $\hat Y^I_{ik} = \hat Y^{1}_{ik}$
		\caption{Paik's sequential mean imputation}
		\label{completed_data}
	\end{algorithm}
	
	The above imputation requires $M(M-1)/2$ estimated models. Then estimating equation (\ref{GEE}) is  solved using the completed data:  
	\begin{equation}\label{GEE_Paik}
		\displaystyle\sum_{i=1}^{N}U_i(\hat L^I_{iM})=\displaystyle\sum_{i=1}^{N} \frac{\partial \mathbf{\upmu}^T_i}{\partial \mathbf{\beta}} V^{-1}_i   (\mathbf{\hat Y}^{I}_i - \mathbf{\upmu}_i) = 0.
	\end{equation}
	
	Under regularity conditions, as long as the mean models $m_k^s$ are correctly specified and the data are MAR, this procedure gives consistent estimates of the parameters $\mathbf{\beta^*}$ \citep{ChoPaik1997} .
	
	\section{Doubly-robust estimators for longitudinal data with dropout}
	\subsection{Optimal longitudinal AIPW estimating equations}    
	In the longitudinal setting with known monotone MAR dropout mechanism, 
	\cite{Tsiatis2006}  
	showed that under regularity conditions, any consistent and asymptotically normal  estimator $\mathbf{\hat \beta}$ of $\mathbf{\beta^*}$ using the observed data solves  AIPW estimating equations of the form  
	\begin{equation} \label{Tsiatis}
		\displaystyle\sum_{i=1}^{N} \bigg( \frac{C_{i,M}}{\pi_{i,M}}U_i(\mathbf{\hat \beta}, \bar L_{iM}) + \displaystyle\sum_{j=1}^{M-1} \bigg(\frac{C_{ij} - \lambda_{i,j+1} R_{ij}}{\pi_{i,j+1}}\bigg) H^{j}(\mathbf{\hat \beta}, \bar L_{ij}) \bigg) = 0
	\end{equation}
	where $H^{j}$  is an arbitrary function.  
	The choice 
	\begin{equation} \label{optimal}
		H^{j}(  \mathbf{\beta}, \bar L_{ij})=E(U( \mathbf{\beta}, \bar L_{iM})|\bar L_{ij}, R_{ij}=1)
	\end{equation}
	yields the estimator with the smallest variance; thus (\ref{Tsiatis}) and (\ref{optimal}) give the optimal observed data estimating equations.  
	More generally, when $\lambda_{ij}$ (and thus $\pi_{ij}$) is estimated by maximum likelihood, and the $\hat H^{j}$ in equation (\ref{optimal}) are estimated by  corresponding models, the solution $\mathbf{\hat \beta^{AIPW}}$ to the optimal estimating equations has the following properties  \citep{Tsiatis2006} :

	\begin{enumerate} 
		\item $\mathbf{\hat \beta^{AIPW}}$ is consistent for $\mathbf{\beta^*}$ and asymptotically normal  if either the dropout models $\hat \lambda_j$  or the imputation models $\hat H^{j}$ are correctly specified and thus consistent for the true conditional expectations they aim to estimate.
		\item If both sets of models are correctly specified, $\mathbf{\hat \beta^{AIPW}}$ has the  smallest asymptotic variance among all doubly-robust estimators of $\mathbf{\beta}$.  
		%
		%
		\item Improved, but more complex, doubly-robust estimators have been proposed, which also attain the minimum asymptotic variance when the imputation models are misspecified but the dropout models are correctly specified and $\hat \lambda$ is an effcient estimator \citep{Tsiatis2011}.  We do not consider these estimators here.
		
	\end{enumerate}
	\subsection{Seaman and Copas' doubly-robust AIPW estimator for longitudinal GEE} 
	For the case where $U_i(\mathbf{\beta})$ is a GEE of the form (\ref{GEE}),     \cite{Seaman2009} proposed a two-step procedure: first obtain estimates of the $\hat H ^{j}$, and then substitute into  the corresponding AIPW estimating equations (\ref{Tsiatis}). In particular, for subject $i$ with $j \leq J_i$,  take
	\begin{equation} \label{DRGEE}
		\hat H^{j}(\mathbf{\hat \beta}, \bar L_{ij}) = \frac{\partial \mathbf{\upmu}^T_i}{\partial \mathbf{\hat \beta}} V^{-1}_i (\mathbf{\hat Y^j_i} - \mathbf{\upmu}_i)
	\end{equation}
	where $\mathbf{\hat Y^j_i}$ is imputed from Paik's sequential mean imputation as in Algorithm \ref{completed_data}.  For $j >J_i$, $H^{j}$ can be taken to be 0, as the weights for $H^{j}$ in 
	(\ref{Tsiatis}) are 0.  Then, the Newton-Raphson algorithm is  used to solve equations (\ref{Tsiatis}) to obtain  $\mathbf{\mathbf{\hat \beta}}$.

	\subsection{Bang and Robins'  doubly-robust recursive regression estimator for $E(Y_{iM})$, and extensions}\label{bang_robbins}  
	\cite{Bang2005} introduced a regression-based approach for longitudinal data with dropout in the particular case where the estimand of interest is $E(Y_{iM}).$  
	The estimator, given in Algorithm \ref{BRobins}, uses recursive regression to impute the values of $Y_{iM}$, including $\hat \pi_{ij}^{-1}$ as a covariate in each imputation model in order to achieve double robustness.
	The algorithm differs from Paik's sequential mean imputation in that it uses imputed values $\hat Y_{iM} $ as the outcome in each estimation step, rather than as predictors, and uses $\hat Y_{iM} $ even when observed values $Y_{iM}$ are available.
	The estimator is shown to be asymptotically equivalent to the doubly-robust AIPW estimator  from  equation (\ref{Tsiatis}) in the cross-sectional setting \citep{Bang2005}. Although \cite{Bang2005} only provided the case of $U_i = Y_{iM} - E(Y_{iM})$, they noted that this approach can be generalized by modifying the $U_i$ and its corresponding score equation.
	\begin{algorithm}
		\KwResult{Bang and Robins doubly-robust estimator $\hat \mu^{BR}_{iM}$ for $E(Y_{iM})$ where $U_i(\bar L_{iM}) = Y_{iM} - E(Y_{iM})$}
		\textbf{Preliminary step:} Estimate $\hat \pi_{i2},...,\hat \pi_{iM}$ by maximum likelihood.
		
		\textbf{Initialize:} Identify all subjects $i$ with $R_{iM} = 1.$ Use these data to regress the observed values of $Y_{iM}$ on $\bar{L}_{i,M-1}$ and $\hat \pi_{iM}$, to obtain a consistently estimated  model $\tilde m^{M-1}_{M}(\bar L_{i,M-1}, \hat \pi_{iM}).$
		For all subjects with $R_{i,M-1} = 1$ let
		\begin{equation*}
			\tilde Y^{M-1}_{i,M} = \tilde m^{M-1}_{iM}(\bar L_{i,M-1}, \hat \pi_{i,M}).
		\end{equation*}
		
		\textbf{For $\mathbf{s = M-2,...,1}$ recursively:}\\
		Identify all subjects $i$ with $R_{i,s+1} = 1$ and regress the values of $\tilde Y^{M-1}_{i,M}$ on $\bar{L}_{is}$ and $\hat \pi_{i,s+1}$, to obtain a consistently estimated model $\tilde m^{s}_{iM}(\bar L_{is}, \hat \pi_{i,s+1})$.  For all subjects with $R_{is} = 1$ let
		\begin{equation*}
			\tilde Y^{s}_{i,M} = \tilde m^{s}_{iM}(\bar L_{i,s}, \hat \pi_{i,s+1}).
		\end{equation*}
		
		\textbf{Final step:}
		Let $ \hat Y^{BR}_{iM} = \tilde Y^{1}_{iM}.$  Substitue  $ \hat Y^{BR}_{iM}$ into the full data estimating functions $U$ and solve to obtain $\hat \mu^{BR}_{iM} = (1/N)\sum_{i=1}^N \hat Y^{BR}_{iM}.$
		\caption{Bang and Robins doubly-robust estimator $\hat \mu^{BR}_{iM}$ for $E(Y_{iM})$}
		\label{BRobins}
	\end{algorithm}
	
	\cite{Schnitzer2016} extended the Bang and Robins approach  to estimate the coefficients $\beta$ giving the association between $Y_{iM}$ and  baseline covariates $X_{i0}$ of interest,   in the context of a generalized linear model. The idea is to take $U_i = X_{i0}(Y_{iM} - g( \mathbf{\beta} X_{i0})),$ with $g$ an appropriate link function and where  $X_{i0}$ includes an intercept.  The model  for $\tilde m^s (\cdot)$ is also taken to use the same link function $g( \cdot)$. Then the approach is similar to  Algorithm \ref{BRobins}, however rather than including $\hat \pi_{i,s+1}$ as a predictor in the model for $\tilde m^s$ , one  includes $\hat \pi_{i,s+1}^{-1} X_{i0}$. Finally, substitute $\tilde Y^{1}_{iM}$   into the estimating functions $U_i$, and solve to obtain the estimate of $\mathbf{\hat \beta}$.  
	For the special case $U_i(\bar L_{iM}) = Y_{iM} - E(Y_{iM}),$ this recovers the 
	\cite{Bang2005} estimator.  
	
	We note here that, when applying the Bang and Robins approach for longitudinal data to estimate $E(Y_{iM})$ , 
	\cite{Tsiatis2011} used forward selection to select variables in both the logistic regression models for estimating $\lambda_j$ and in the OLS models for estimating $E(Y_{iM}|\bar L_{ij})$.  Similarly, in their implementation Schnitzer and colleagues performed a variable selection step in each regression.  In our simulations we incorporate a similar variable selection step for these estimators, as it greatly improved their performance  in practice.  We denote this slightly modified and extended estimator {\bf BR*}. Also note that the estimator requires the sequential regression models $\tilde m^s$ in Algorithm \ref{BRobins} to be consistent with the estimating equation $U$ that defines the estimand of interest; thus a change of estimating equation may require a new set of sequential regressions.  This seems to be a key difference from the imputation framework proposed below.

	\section{An imputation framework for longitudinal data using AIPW estimators}
	
	Here we show that the optimal AIPW estimating equations (\ref{Tsiatis})  can be written in a form that applies the full data estimating functions $U(\mathbf{\beta}, \bar L_{iM})$ to the 'completed' data, in which all levels of missing observations have been filled in using a set of doubly-robust imputed observations. In this approach $U(\mathbf{\beta}, \bar L_{iM})$ may be any equation of interest linear in $\mathbf{Y}_i$, including   a GEE.  The approach has the advantage that standard software can be used to solve for the estimates  $\mathbf{\hat\beta^{AIPW}},$ and it provides a framework for  flexible construction of AIPW estimators. In particular, unlike other doubly-robust approaches, when the parameters of interest change there is no need to run the  doubly-robust estimating procedure again.
	
	\subsection{The optimal AIPW estimating equations in imputation form}
	Consider the estimating functions $U(\mathbf{\beta}, \bar L_{iM}) = q(\mathbf{\beta})(\mathbf{Y}_i - \mathbf{\upmu}_i)$, where $\mathbf{Y}_i = (y_{i,1}, y_{i,2},..., y_{i,M})^T$ and $\mathbf{\upmu}_i = (\mu_{i,1}, \mu_{i,2},..., \mu_{i,M})^T$; for example  $q(\mathbf{\beta}) = \frac{\partial \mathbf{\upmu}_i^T}{\partial \mathbf{\beta}} V_i^{-1}$ with $\mu(\mathbf{\beta}) = g(\mathbf{\beta} X)$ when $U(\mathbf{\beta}, \bar L_{iM})$ is a GEE. Then one can show that the optimal AIPW estimator given by equations (\ref{Tsiatis}) and (\ref{optimal})  has an equivalent  substitution form:
	\begin{equation}\label{UDR}
		\displaystyle\sum_{i=1}^{N}U(\mathbf{\hat \beta}, \bar L^{AIPW}_{i,M}) =  
		\displaystyle\sum_{i=1}^{N} q(\mathbf{\hat \beta}) (\mathbf{\hat Y}^{AIPW}_i - \mathbf{\upmu}_i) = 0
	\end{equation}
	where $\mathbf{\hat Y}^{AIPW}_i $ is a corresponding AIPW estimator of the full data $\mathbf{Y}_i$, given in (\ref{hatY}) below.  Specific examples of such AIPW imputation estimators are given in Section \ref{sec6}.

	To demonstrate (\ref{UDR}), first, it is straightforward to show that 
	\begin{eqnarray}\label{sec5.1}
		\frac{C_{i,M}}{\pi_{i,M}} + \displaystyle\sum_{j=1}^{M-1} \bigg(\frac{C_{ij} - \lambda_{i,j+1} R_{ij}}{\pi_{i,j+1}}\bigg) = 1.
	\end{eqnarray}
	Thus, letting $ \pi_{i,M+1}=\pi_{iM}$ and $\lambda_{i1}=\lambda_{i,M+1}=0$, equation (\ref{Tsiatis}) can be rewritten as
	\begin{equation} \label{Tsiatis2}
		\displaystyle\sum_{i=1}^{N} \displaystyle\sum_{j=1}^{M} \bigg(\frac{C_{ij} -   \lambda_{i,j+1} R_{ij}}{ \pi_{i,j+1}}\bigg)   H^{j}(\mathbf{\beta}, \bar L_{ij}) =0 .
	\end{equation}

	Next,  following 
	\cite{Seaman2009} (for the case where $q(\mathbf{\beta}) = \frac{\partial \mathbf{\upmu}_i^T}{\partial \mathbf{\beta}} V_i^{-1}$), we have   that (\ref{optimal}) can be written as
	\begin{equation*} \label{Hjefficient}
		H^{j}(\mathbf{\beta}, \bar L_{ij}) = q(\mathbf{\beta})(E[\mathbf{Y}_i |\bar{L}_{ij}]- \mathbf{\upmu}_i).
	\end{equation*}
	Then substituting into (\ref{Tsiatis2}), switching the order of summation, and recognizing that 
	$\sum_{j=1}^{M}  (C_{ij} - \hat \lambda_{j+1} R_{ij})/{\hat \pi_{j+1}} =1  $,
	we obtain that (\ref{Tsiatis}) can be written as:
	\begin{align}\label{efficientDR}
		q(\mathbf{\beta}) \bigg\{\big( \displaystyle\sum_{j=1}^{M} (\frac{C_{ij} - \lambda_{i,j+1} R_{ij}}{\pi_{i,j+1}})E[\mathbf{Y}_i |\bar{L}_{ij}, R_{ij} = 1] \big) - \mathbf{\upmu}_i \bigg\}.
	\end{align}
	Finally, we can write
	\begin{align}\label{hatY}
		\mathbf{\hat Y}^{AIPW}_i=  \displaystyle\sum_{j=1}^{M} (\frac{C_{ij} - \hat \lambda_{i,j+1} R_{ij}}{\hat \pi_{i,j+1}}) E[\mathbf{\hat Y}_i |\bar{L}_j, R_j = 1] 
	\end{align}
	to obtain (\ref{UDR}), recognizing $\mathbf{\hat Y}^{AIPW}_i $ as the representation of an optimal AIPW estimator of $\mathbf{\beta} = E[\mathbf{Y}_i].$  
	
	Form (\ref{UDR}) has the advantage over the original approach in \cite{Seaman2009}  that, once the values of $\mathbf{\hat Y}^{AIPW}_i $ are obtained,  standard software for the estimating functions $U(\mathbf{\beta}, \bar L_{iM})$ can be used to solve for the doubly-robust estimator $\mathbf{\beta^*}$.  Furthermore, the form of $\mathbf{\hat Y}^{AIPW}_i$ is independent of the details of the estimating functions $U$, an independence that we can exploit to address additional estimands defined by different estimating functions, using the same set of values $\mathbf{\hat Y}^{AIPW}_i .$  From the theory in 
	\cite{Tsiatis2006}, we can be assured that $\mathbf{\beta^*}$  is a member of the class of locally efficient DR estimators.  
	We may also use the above considerations to provide a direct  demonstration that a general estimator $\mathbf{\hat \beta^{AIPW}}$  is doubly-robust. A formal development of this proof can be found in Appendix \ref{secApp1}.
	
	%
	%
	
	\section{Imputation framework with two specific DR estimators}\label{sec6} 
	Here we apply the AIPW imputation framework to construct two particular DR estimators for longitudinal data with monotone dropout.

	\subsection{AIPW-I: sequential mean imputation}\label{AIPW} 
	We denote the sequential mean imputation  implementation of the DR imputation estimator as  $\mathbf{\hat \beta^{AIPW-I}}.$ 
	\vspace{1em}
	The procedure can be described in two steps: 
	\vspace{1em}

	\begin{description}
		\item[\em Imputation step:] For each subject $i$, impute a doubly-robust complete data vector using the AIPW estimate of $Y_{ik}$, $k= 2,\ldots M$
		\begin{equation} \label{imp_AIPW}
			\hat{Y}_{ik}^{AIPW}=\frac{C_{ik} Y_{ik}}{\hat \pi_{ik}} + \displaystyle\sum_{j=1}^{k-1} (\frac{C_{ij} - \hat \lambda_{ij+1} R_{ij}}{\hat \pi_{ij+1}}) \hat{m}_{k}^{j} (\bar L_{ij}) 
		\end{equation}
		with the models $\hat m^j_k(\cdot) = E[Y_{ik} | \bar L_{ij}, R_{ij}=1]$ estimated by Paik's sequential mean regression as in Algorithm \ref{completed_data}.
		\vspace{1em}
		\item[\em Estimation step:] Substitute $\hat{\mathbf{Y}}_{i}^{AIPW}$  into the full data estimating equations. These equations are  then  solved using standard software to obtain the doubly-robust estimator $\mathbf{\hat \beta^{AIPW-I}}.$
	\end{description}

	\subsection{AIPW-S:  a computationally simpler baseline $\times$  time imputation model}\label{AIPW-S}
	Our second implementation is computationally simpler,  denoted as $\mathbf{\hat \beta^{AIPW-S}}.$  It uses only baseline covariates and time for the imputation models, at the cost of potential efficiency loss and  stronger required assumptions. This is motivated by the clinical trials setting, in which it may be desirable to restrict attention to modeling with pre-randomization covariates.  For the imputation models, it uses up to $M-1$ models with baseline covariates and time, rather than the  $M(M-1)/2$ sequential recursive regression models required by 
	$\mathbf{\hat \beta^{AIPW-I}}$ or the Bang and Robins estimator.  It can be thought of as a simplified version of the fully efficient approach given in 
	\cite{Tsiatis2011}.  If the covariates are sufficient to render the data MAR, then $\mathbf{\hat \beta^{AIPW-S}}$  will be doubly-robust. The cost is potential efficiency loss, and a stronger assumption regarding the MAR conditions. In the simulations and application, we give an example of the simplest implementation, in which only one single MMRM is estimated rather than $M-1$ models, or $M(M-1)/2$ models as in the AIPW-I and related estimators.
	
	Here we give a sketch of the development; please refer to the Appendix \ref{secapp2} for details. To develop the estimator, we start with (\ref{efficientDR}), and note that $\hat E[Y_{ik} |\bar L_{ij}] = Y_{ik}$ for $j \geq k$. Hence we can write the $k^{th}$ component of $\mathbf{\hat Y}_i^{AIPW}$ as
	\begin{align}\label{YAIPWS}
		\hat Y^{AIPW}_{ik} &=
		\displaystyle\sum_{j=k}^{M} (\frac{C_{ij} - \hat \lambda_{i,j+1} R_{ij}}{\hat \pi_{i,j+1}}) Y_{ik}
		+
		\displaystyle\sum_{j=1}^{k-1} (\frac{C_{ij} - \hat \lambda_{i,j+1} R_{it}}{\hat \pi_{i,j+1}}) \hat Y_{ik}^{j}.
	\end{align}
	
	Next, we take $\hat Y^{j}_{ik}$ to be a consistent estimate of $E[Y_{ik}|X_{i0}, t]$ independent of $j$, where $X_{i0}$ contains baseline covariates and $t$ indicates time. This amounts to the choice $H^{j}(  \mathbf{\beta}, \bar L_{ij})=E[Y_{ik}|X_{i0}, t]$.  Note that, while we are now outside the set of possible efficient estimators given by   equation (\ref{Tsiatis}), except in the special case where $E[Y_{ik}|\bar L_{ij}] = E[Y_{ik}|X_{i0}, t]$, the arguments for double robustness remain unchanged. (Please see Appendix  \ref{secapp2} for details.) %

	In this case, (\ref{YAIPWS}) simplifies to a repeated cross-sectional form:
	\begin{equation}\label{BL}
		\hat Y^{AIPW-S}_{ik}  = \frac{R_{ik}}{\hat \pi_{ik}} Y_{ik} +
		(1-\frac{R_{ik}}{\hat \pi_{ik}}) \hat Y_{ik},
	\end{equation}
	because $ \sum_{j=k}^{M} ((C_{ij} - \hat \lambda_{i,j+1} R_{ij})/\hat \pi_{i,j+1}) = R_{ik}/\hat \pi_{ik}$.

	\vspace{1em}The procedure can be described in two steps:
	\begin{description}
		\item[\em Imputation step:] For each subject $i$, impute a doubly-robust complete data vector using the AIPW-S estimate of $Y_{ik}$, $k= 2,\ldots M$ given by (\ref{BL}), with  $\hat Y_{ik} = \hat m_k(X_{i0}, t)$ depending only on baseline covariates $X_{i0}$ and time $t,$ and where $\ m_k(X_{i0}, t)= E[Y_{ik}|X_{i0}, t]$ is estimated from the observed data.
		\vspace{1em}
		\item[\em Estimation step:] Substitute $\hat{\mathbf{Y}}_{i}^{AIPW-S}$  into the full data estimating equations. These equations are  then  solved using standard software to obtain the doubly-robust estimator $\mathbf{\hat \beta^{AIPW-S}}.$
	\end{description}

	A detailed development can be found in Appendix \ref{secapp2}. In the simulation and applications below, we give  examples where  a single mixed-effects model $\hat m(X_0,t)$  is estimated using all observed responses as the outcomes, regressed on baseline covariates $X_0$ and time $t$.  Then $ \hat Y_{ik} = \hat m(X_{i0},t)$, and  formula (\ref{BL}) is used to derive a completed doubly-robust data set. In this case, instead of $M(M-1)/2$ regressions as for the other sequential regression estimators, only one regression needs to be estimated.

	
	\section{Simulations}\label{sec:Simulation}
	We use simulation to investigate the performance of the two  DR imputation-framework estimators  in the setting of normal longitudinal generalized estimating equations under MAR monotone dropout:  {\bf  AIPW-I,} based on AIPW using Paik's sequential mean imputation (section \ref{AIPW}); and  {\bf AIPW-S}, AIPW with a computationally simpler imputation model using only baseline covariates (section \ref{AIPW-S}).  We also compare the use of the Bang and Robins regression-based estimator {\bf BR*} in both its original form and as extended by  \cite{Schnitzer2016} to more general estimating equations for generalized linear models. 	
	
	We study two different estimands, in two sets of simulations: 
	first, $E[Y_{k}]$ for the case $k=1, \ldots M$, in which case {\bf BR*} reduces to the original Bang and Robins estimator,  and  second, the vector of regression coefficients $\mathbf{\beta}$ defined as the solution (using the full-data distribution) to the score equations $U$ given by a mixed model with repeated measures (MMRM), similar to what might be used in the primary analysis of a clinical trial.  
	In particular, $\mathbf{\beta}$ includes a time by treatment interaction,  which represents an important estimand in the clinical trials setting.  
	
	For comparison, we also include two traditional observed-data regression estimators often used in clinical trials: (1) a mixed model of repeated measures ({\bf MMRM}) with random intercept and slope, and (2) a generalized estimating equations model using an independence working correlation ({\bf GEE-IND)} with a sandwich estimator of variance. As negative controls which adjust for dropout but are not DR, we also include Paik's sequential mean imputation ({\bf Paik}), as described in section \ref{Paik_algorithm}, and  inverse probability weighted GEE ({\bf WGEE}) as implemented in the R package wgeesel \citep{wgee}.  Finally, as a further negative control in the imputation framework, we use  data imputed using {\bf BR*}  in the first simulation, for estimating  a time by treatment interaction in the second simulation.  As  {\bf BR*} is not designed to work as an imputation estimator, under model misspecification this may  demonstrate  the utility of the new estimators in practice, in a special case.  Notably, we allow the imputation model to differ from the estimation model, in order to show the utility of the imputation framework.
	
	The two sets of simulations are linked for the imputation estimators.  We used the estimates of  $E[Y_{k}]$ from the first set of simulations to obtain a completed dataset.  We then used a standard GEE model with the completed data to estimate the regression coefficients $\beta$ in the second set of simulations.	 	
	
	For each simulated data set and estimation method, we compute both a point estimate for the estimand of interest and an associated 95\% normal-theory confidence interval, where the variance of the parameter is computed using the nonparametric bootstrap.
	We consider both a moderate dropout construct  ($\approx 30\%$ dropout) and an extreme  dropout construct  ($\approx 50$\% dropout). 
	To investigate double robustness, we construct four scenarios, depending on whether the dropout models and/or the imputation models are specified correctly or incorrectly.  
	
	Importantly, we  include in the study an extreme form of model misspecification as described in 7.2 and 7.3, in which the incorrect dropout models or imputation models do not include the treatment indicator, and the estimand is the coefficient of a treatment by time interaction.  This case allows the estimand analysis model to differ substantially from the imputation model, and illustrates what might happen when using a pre-existing imputed complete dataset to study a different estimand of interest.
	
	\subsection{Performance metrics and sample sizes}
	We report Monte Carlo estimates (500 repeats) for the bias, standard deviation, and root mean square error (RMSE) of the point estimate. For the confidence interval, we report
	the coverage probability, the mean of the bootstrap standard error estimate, and the mean interval score  of 
	\cite{Gneiting2007}, given by:
	\begin{equation} \label{score}
		S(\hat l , \hat u, \theta) = (\hat u -\hat l) + \frac{2}{\alpha} \left (  (\hat l- \theta ) \mathds{1} \{   \theta < \hat l \} +(\theta - \hat u) \mathds{1} \{  \hat u\  <\theta \} \right ),
	\end{equation}
	where $\theta$ is the true parameter of interest, $1-\alpha$ is the nominal confidence level, $(\hat l, \hat u)$ are the interval limits. and $\mathds{1}\{\}$ denotes the indicator function. A MMRM model with correct covariates and random intercept and slope will serve as the gold standard in our comparisons; this reflects the complete data generating model.


	\subsection{Specification of the data generating model, the primary estimand, and the correct and incorrect imputation  models}
	Longitudinal responses $Y_{ij}$ are generated from  a mixed-effects model, using similar parameters as in 
	\cite{Tsiatis2011}, as
	\begin{eqnarray}\label{regression}
		Y_{ij} = & b_{0i} + b_{1i} t + \beta_0 + \beta_1 x_1 + \beta_2 x_2 + \beta_3 x_2 \times t + \epsilon.
	\end{eqnarray}
	The covariates are generated as $x_1 \sim$ N(5,1),  $x_2 \sim$   Bernoulli$(0.5)$, with $\epsilon \sim$ N(0, 1). Considering the  clinical trial setting, $x_2$ would indicate the treatment variable, and $x_1$ a continuous covariate.  The sample size is $n=500$, and three time points are $t={1,2,3}$.  The bootstrap sample size (for variance estimation) is 300. 
	Here,  $b_{0}$ and $b_1$ are random intercepts and  slopes  from a bivariate normal with mean 
	$\mu = \big(\begin{smallmatrix}
		1\\
		6
	\end{smallmatrix}\big)$ 
	and covariance 
	$\Sigma = \big(\begin{smallmatrix}
		0.3 & 0.1\\
		0.1 & 0.2
	\end{smallmatrix}\big)$.  The data generating coefficients are $\beta_0 = 0.5$, $\beta_1 = 2$, $\beta_2 = -0.25$, $\beta_3 = -6$.  Thus the  expectation of $Y_{ij}$ is $11.375, 14.375, 17.375$  for $j= 1,2,3$. 
	
	Only a few studies have performed longitudinal simulations to evaluate AIPW estimators. Comparing with  
	\cite{Tsiatis2011} and others, we further include an arm by time interaction, providing a more realistic model for randomized trials. The primary estimand is either $E[Y_3]$ or the vector of regression coefficients $\mathbf{\beta}_{est} = (E(b_{1i}), \beta_2, \beta_3)$ obtained from the estimating equations for model (\ref{regression}) , referring to the slope of time, difference between arms at baseline, and effect of treatment along with time.
	
	For all methods, the correctly specified mean model includes categorical time (if needed),  $(x_1, x_2),$ and the interaction of time by $x_2$.  For sequential imputation methods, the correct  models are of the form $Y_j \sim Y_1 + ... + Y_{j-1} + x_1 + x_2 + e$ for $j=(2,3)$.  For all methods, the misspecified imputation model   excludes the treatment indicator $x_2$ and the corresponding interaction terms. Note that in our setting this is a case of severe model mis-specification, as the difference between treatment arms is of primary interest. From another perspective, there is a severe mis-match between the imputation models $m,$ which omit the variable of primary interest, and  the estimating equations $U$ which define the estimand.  Thus, this may also be thought of as  a test of the robustness of the imputation framework to changing the estimand of interest after the imputation step has been carried out. Details of the model specifications can be found in the Appendix \ref{apend3}.
	\subsection{Specification of the dropout generating model and of correct and incorrect  models for missingness}
	
	We generate dropout according to the logistic regression models   
	\begin{align*} 
		logit(\lambda_{2}) & = \gamma_{20} + \gamma_{21} y_1 - \gamma_{22} x_2 + e\\ 
		logit(\lambda_{3}) & =  \gamma_{30}+ \gamma_{31}  y_1 + \gamma_{32}y_2 - \gamma_{33}  x_2 + e,
	\end{align*}
	where
	$\lambda_{ij}$ is the hazard function of dropout as defined in formula (\ref{HZfunc}). 
	For the moderate dropout construct, 
	$(\gamma_{20}, \gamma_{21}, \gamma_{22}, \gamma_{30}, \gamma_{31}, \gamma_{32}, \gamma_{33} )  = (-7.625, 0.5, 2, -5.225, 0.1, 0.2, 4)$. 
	The empirical mean dropout rates from 500  Monte Carlo repeats were 10\% for $Y_2$ and 30\% for $Y_3$.  For the high dropout construct, $(\gamma_{20}, \gamma_{21}, \gamma_{22}, \gamma_{30}, \gamma_{31}, \gamma_{32}, \gamma_{33} )  = (-7, 0.5, 1, -4.5, 0.1, 0.2, 2)$ and the mean dropout rates were 20\% and 48\% for $Y_2$ and $Y_3,$ respectively.
	Misspecified dropout models omit the treatment indicator $x_2$, which again omits the difference between treatment arms in the clinical trials setting. 
	
	\subsection{Specification of the {BR*} estimator, used as an imputation estimator} \label{BR}
	In the case when the estimand of interest is the coefficient of  a baseline covariate $X$, and where the imputation models are consistent with the estimating equations,  the { BR*} estimator, given in section (\ref{bang_robbins}), is similar in form to an imputation framework estimator.    Hence we  investigated empirically the use of BR* in the imputation step (\ref{imp_AIPW}).  This can be done by setting $M$ in Algorithm \ref{BRobins} sequentially from $M= 2$ to the final visit, and then substituting each $\hat Y_{ik}^{AIPW}$ in (\ref{imp_AIPW})  by $\hat \mu^{BR}_{ik}.$ 
	Note that, in the case of estimating the coefficient of a treatment by time interaction term,  we are outside of the category of estimands considered by \cite{Bang2005} and \cite{Schnitzer2016}.  Also, the imputation framework allows the model in the imputation step to differ from the estimating equations model, and so in this case is again outside the class of estimators considered in those papers.  In our simulations, we  used forward variable selection in the imputation models, following the implementations in both \cite{Tsiatis2011} and \cite{Schnitzer2016}, as this in practice improved performance.

	\subsubsection{Results for estimating $E(Y_3),$ under moderate dropout}
	
	The upper left panel of Table \ref{table1 Moderate Scenario} shows results for estimating $E(Y_3)$ under the moderate dropout construct when both imputation and dropout models are correct. The "gold standard" MMRM model (correct complete data maximum likelihood estimator, with MAR data) has a bias of -0.01. GEE-IND has a worse bias of -0.09, which supports the theory that even if the mean structure is correct, an incorrect working correlation may still cause bias under MAR longitudinal dropout.  The bias of all three doubly-robust methods is less than 0.01, and the performance is very similar to the gold standard, consistent with their asymptotic local efficiency. Paik's imputation and WGEE also perform well in terms of bias and efficiency. 
	However, across all moderate dropout scenarios, approximately 1\% of  Monte Carlo repeats for WGEE reported convergence issues, sometimes resulting in substantial standard errors and outlier estimates. The program also consumes much more time than other methods, by a factor of about 5.   
	
	When the imputation model is correct but the dropout model is misspecified (upper right panel), all doubly-robust methods have acceptable bias, ranging from 0.037 to -0.005, with RMSEs, coverage probabilities,  interval scores, Monte Carlo standard deviations, and average estimated standard errors all similar to the gold standard, indicating the estimators are consistent and remain close to efficient. 
	
	When the dropout model is correct but the imputation model is misspecified (bottom left panel), all doubly-robust methods again have acceptable bias, indicating their double robustness. AIPW-I has the best efficiency, and AIPW-S loses some efficiency compared to prior scenarios. 
	In this scenario, the original BR estimator without variable selection had an unacceptably large bias and low coverage probability (85\%, not shown), although the modified estimator BR* (from 
	\cite{Tsiatis2011}) performs well, as shown in the table \ref{table1 Moderate Scenario}.
	As expected, the non-doubly-robust regression-based methods did not work well. The incorrect MMRM model and Paik's imputation have coverage probabilities less than 50\%, with large bias and bad efficiencies. A GEE model with wrong mean structure and wrong working correlation matrix performed the worst.
	
	The bottom right panel shows results when both models are incorrect. All methods have worse performance than in other scenarios, indicating that the misspecification in the models is substantial, and thus provides a good test of double robustness. In this scenario, all methods report similar bias and efficiencies except GEE-IND, which reports a much worse result.
	
	\begin{table*}[!t]
		\caption{Comparison by simulation of point and interval estimators for $E(Y_3),$ under a moderate dropout rate of 30\% missing.    Performance metrics are: Bias, Root mean square error (RMSE), the interval score (IntS), coverage probability (CovP), Monte Carlo standard deviation (MCSD) and mean bootstrap estimate of standard error (AveSE), from 500 simulation runs.  The sample size was n=500, and the bootstrap sample size was 300. The true value of $E(Y_3)= 17.4$. 
		}
		\label{table1 Moderate Scenario} 
		\renewcommand{\arraystretch}{1.1}
		\resizebox{\textwidth}{!}{
			\begin{tabular}{l|rrrrrrrrrrrrr}
				\hline
				\hline
				\\
				& Bias & RMSE & IntS & CovP & MCSD & AveSE & & Bias & RMSE & IntS & CovP & MCSD & AveSE \\ 
				\hline
				& \multicolumn{6}{c}{\underline{Y correct P correct}} & & \multicolumn{6}{c}{\underline{Y correct P incorrect}} \\
				BR* & -0.01 & 0.30 & 1.37 & 0.95 & 0.30 & 0.30 &  & -0.01 & 0.30 & 1.37 & 0.95 & 0.30 & 0.30 \\ 
				AIPW-I & -0.01 & 0.30 & 1.39 & 0.95 & 0.30 & 0.31 &  & -0.00 & 0.30 & 1.37 & 0.95 & 0.30 & 0.30 \\ 
				AIPW-S & -0.01 & 0.31 & 1.39 & 0.95 & 0.31 & 0.31 &  & 0.04 & 0.31 & 1.39 & 0.95 & 0.31 & 0.31 \\ 
				Paik & -0.01 & 0.30 & 1.35 & 0.95 & 0.30 & 0.30 &  & -0.01 & 0.30 & 1.35 & 0.95 & 0.30 & 0.30 \\ 
				MMRM & -0.01 & 0.30 & 1.38 & 0.95 & 0.30 & 0.30 &  & -0.01 & 0.30 & 1.38 & 0.95 & 0.30 & 0.30 \\
				WGEE & 0.00 & 0.30 & 1.37 & 0.95 & 0.30 & 0.30 &  & 0.03 & 0.31 & 1.39 & 0.95 & 0.31 & 0.31 \\ 
				GEE-IND & -0.09 & 0.31 & 1.40 & 0.94 & 0.29 & 0.30 &  & -0.09 & 0.31 & 1.40 & 0.94 & 0.29 & 0.30 \\ 
				\hline
				& \multicolumn{6}{c}{\underline{Y incorrect P correct}} & & \multicolumn{6}{c}{\underline{Y incorrect P incorrect}} \\
				BR* & -0.00 & 0.33 & 1.54 & 0.95 & 0.33 & 0.34 &  & -0.64 & 0.71 & 6.06 & 0.53 & 0.31 & 0.33 \\ 
				AIPW-I & -0.01 & 0.31 & 1.44 & 0.95 & 0.31 & 0.31 &  & -0.68 & 0.74 & 7.65 & 0.42 & 0.31 & 0.31 \\ 
				AIPW-S & -0.04 & 0.38 & 1.72 & 0.95 & 0.38 & 0.36 &  & -0.62 & 0.71 & 6.44 & 0.55 & 0.36 & 0.36 \\ 
				Paik & -0.65 & 0.72 & 6.95 & 0.45 & 0.31 & 0.32 &  & -0.65 & 0.72 & 6.95 & 0.45 & 0.31 & 0.32 \\ 
				MMRM & -0.60 & 0.67 & 5.83 & 0.50 & 0.31 & 0.32 &  & -0.60 & 0.67 & 5.83 & 0.50 & 0.31 & 0.32 \\
				WGEE & 0.00 & 0.32 & 1.44 & 0.94 & 0.32 & 0.32 &  & -0.62 & 0.71 & 6.11 & 0.53 & 0.34 & 0.36 \\  
				GEE-IND & -2.18 & 2.20 & 63.84 & 0.00 & 0.30 & 0.31 &  & -2.18 & 2.20 & 63.94 & 0.00 & 0.30 & 0.31 \\ 
				\hline
				\hline
				
		\end{tabular}}
		\\
		{\tiny

		}
	\end{table*}
	
	\begin{table*}[!b]
		\caption{Comparison by simulation of point and interval estimators for regression coefficients $\mathbf{\beta}$,   under a moderate dropout rate of 30\% missing.    Performance metrics are: Bias, Root mean square error (RMSE), the interval score (IntS), and coverage probability (CovP) for a 95\% confidence interval, from 500 simulation runs.  The sample size was n=500, and the bootstrap sample size was 300. The true values of the parameters $\mathbf{\beta}$ are 2, -0.25, and 6 for the coefficients of $x_2$ (treatment indicator), time, and the $x_2$:time  interaction term, respectively.}
		\label{table2 Moderate Scenario} 
		\renewcommand{\arraystretch}{1.1}
		\centering
		\resizebox{\textwidth}{!}{%
			\begin{tabular}{l|cccccccccccccc}
				\hline
				\hline
				& \multicolumn{4}{c} {coefficient of $x_2$}                         &  & \multicolumn{4}{c} {coefficient of $time$}     &  & \multicolumn{4}{c} {coefficient of $ time \times x_2$}  \\ \cline{2-5} \cline{7-10} \cline{12-15} 
				& Bias                     & RMSE & IntS  & CovP &  & Bias  & RMSE & IntS   & CovP &  & Bias  & RMSE & IntS   & CovP \\ \hline
				& \multicolumn{14}{c}{Y correct P correct}                                                                           \\
				BR*      & 0.01 & 0.13 & 2.24 & 0.98 &  & 0.00 & 0.11 & 0.50 & 0.93 &  & -0.01 & 0.14 & 0.67 & 0.94 \\
				AIPW-I  & 0.01                     & 0.10 & 0.50  & 0.95 &  & 0.00  & 0.11 & 0.50   & 0.94 &  & -0.01 & 0.14 & 0.68   & 0.93 \\
				AIPW-S  & 0.01                     & 0.10 & 0.50  & 0.95 &  & 0.00  & 0.11 & 0.50   & 0.95 &  & -0.01 & 0.14 & 0.68   & 0.93 \\
				Paik    & 0.01                     & 0.10 & 0.50  & 0.95 &  & 0.00  & 0.10 & 0.50   & 0.93 &  & -0.01 & 0.14 & 0.67   & 0.93 \\
				MMRM    & 0.01                     & 0.10 & 0.50  & 0.94 &  & 0.00  & 0.10 & 0.50   & 0.94 &  & -0.01 & 0.14 & 0.66   & 0.94 \\
				WGEE    & 0.01                     & 0.10 & 0.50  & 0.95 &  & -0.01 & 0.10 & 0.50   & 0.95 &  & 0.00  & 0.14 & 0.65   & 0.94 \\
				GEE-IND & 0.01                     & 0.10 & 0.50  & 0.95 &  & -0.03 & 0.11 & 0.51   & 0.94 &  & 0.02  & 0.14 & 0.64   & 0.95 \\ \hline
				& \multicolumn{14}{c}{Y correct P incorrect}                                                                         \\
				BR*      & 0.01 & 0.18 & 1.70 & 0.98 &  & 0.00 & 0.10 & 0.50 & 0.94 &  & -0.01 & 0.14 & 0.66 & 0.94 \\
				AIPW-I  & 0.01                     & 0.10 & 0.50  & 0.95 &  & 0.00  & 0.10 & 0.50   & 0.94 &  & -0.01 & 0.14 & 0.68   & 0.93 \\
				AIPW-S  & 0.01                     & 0.10 & 0.50  & 0.95 &  & -0.01 & 0.10 & 0.50   & 0.94 &  & 0.02  & 0.14 & 0.66   & 0.94 \\
				Paik    & 0.01                     & 0.10 & 0.50  & 0.95 &  & 0.00  & 0.10 & 0.50   & 0.93 &  & -0.01 & 0.14 & 0.67   & 0.93 \\
				MMRM    & 0.01                     & 0.10 & 0.50  & 0.94 &  & 0.00  & 0.10 & 0.50   & 0.94 &  & -0.01 & 0.14 & 0.66   & 0.94 \\
				WGEE    & 0.01                     & 0.10 & 0.51  & 0.95 &  & -0.01 & 0.10 & 0.50   & 0.95 &  & 0.02  & 0.14 & 0.65   & 0.95 \\
				GEE-IND & 0.01                     & 0.10 & 0.50  & 0.95 &  & -0.03 & 0.11 & 0.51   & 0.94 &  & 0.02  & 0.14 & 0.64   & 0.95 \\ \hline
				& \multicolumn{14}{c}{Y incorrect P correct}                                                                         \\
				BR*      & 0.00 & 0.14 & 2.79 & 0.98 &  & -0.83 & 0.89 & 20.82 & 0.04 &  & 1.65 & 1.75 & 61.60 & 0.00 \\
				AIPW-I  & 0.01                     & 0.10 & 0.50  & 0.95 &  & 0.00  & 0.11 & 0.53   & 0.95 &  & -0.01 & 0.15 & 0.70   & 0.94 \\
				AIPW-S  & 0.01                     & 0.11 & 0.51  & 0.95 &  & 0.00  & 0.11 & 0.55   & 0.95 &  & -0.01 & 0.15 & 0.71   & 0.94 \\
				Paik    & 0.01                     & 0.10 & 0.51  & 0.95 &  & -0.59 & 0.61 & 13.43  & 0.01 &  & 0.67  & 0.69 & 15.12  & 0.01 \\
				MMRM    & 0.25                     & 0.25 & 10.00 & 0.00 &  & -3.24 & 3.25 & 117.84 & 0.00 &  & 6.00  & 6.00 & 240.00 & 0.00 \\
				WGEE    & 0.25                     & 0.25 & 10.00 & 0.00 &  & -3.07 & 3.08 & 111.33 & 0.00 &  & 6.00  & 6.00 & 240.00 & 0.00 \\
				GEE-IND & 0.25                     & 0.25 & 10.00 & 0.00 &  & -3.30 & 3.31 & 120.29 & 0.00 &  & 6.00  & 6.00 & 240.00 & 0.00 \\ \hline
				& \multicolumn{14}{c}{Y incorrect P incorrect}                                                                       \\
				BR*      & \multicolumn{1}{l}{0.02} & 0.25 & 3.59 & 0.97 &  & -3.30 & 3.31 & 119.88 & 0.00 &  & 6.08 & 6.08 & 238.57 & 0.00 \\
				AIPW-I  & \multicolumn{1}{l}{0.01} & 0.11 & 0.50  & 0.95 &  & -0.29 & 0.31 & 3.38   & 0.34 &  & 0.05  & 0.15 & 0.67   & 0.95 \\
				AIPW-S  & \multicolumn{1}{l}{0.01} & 0.11 & 0.50  & 0.95 &  & -0.32 & 0.34 & 4.05   & 0.28 &  & 0.10  & 0.18 & 0.79   & 0.91 \\
				Paik    & \multicolumn{1}{l}{0.01} & 0.10 & 0.51  & 0.95 &  & -0.59 & 0.61 & 13.43  & 0.01 &  & 0.67  & 0.69 & 15.12  & 0.01 \\
				MMRM    & \multicolumn{1}{l}{0.25} & 0.25 & 10.00 & 0.00 &  & -3.24 & 3.25 & 117.84 & 0.00 &  & 6.00  & 6.00 & 240.00 & 0.00 \\
				WGEE    & \multicolumn{1}{l}{0.25} & 0.25 & 9.98  & 0.00 &  & -3.27 & 3.27 & 118.58 & 0.00 &  & 6.00  & 6.00 & 239.97 & 0.00 \\
				GEE-IND & \multicolumn{1}{l}{0.25} & 0.25 & 10.00 & 0.00 &  & -3.30 & 3.31 & 120.29 & 0.00 &  & 6.00  & 6.00 & 240.00 & 0.00 \\
				\hline 
				\hline
			\end{tabular}%
		}
	\end{table*}
	
	\subsubsection{Results for estimating coefficients $\mathbf{\beta}$, under moderate dropout}
	Table \ref{table2 Moderate Scenario} presents results for estimating
	the vector of three regression coefficients $\mathbf{\beta} = (E(b_1), \beta_2, \beta_3)$ as defined in (\ref{regression}). 
	The three regression coefficients are the coefficients of $time$, $x_2$,  and $x_2$ by $time$. In a randomized trial, these coefficients would correspond to the  time trend for the placebo arm, the baseline difference between arms, and the treatment effect on the time trend.
	
	With both models correctly specified, the estimated $\mathbf{\hat \beta}$ appears to be consistent and efficient in all  methods, except that the Bang and Robins - Schnitzer approach (BR*) reports a larger variance for the estimated $\beta_2$, the coefficient of the treatment indicator $x_2$. After excluding simulation runs with  substantial standard errors (greater than 50), BR* still has a large mean interval score and some unusual standard errors from Monte Carlo repeats. 
	
	When the imputation model is correctly specified but the dropout model incorrectly excludes the treatment indicator $x_2$, the performance of the three doubly-robust estimators is as good as the first scenario, indicating that all three doubly-robust imputed data sets have minimal bias in this scenario. While AIPW-I and AIPW-S also obtain comparable efficiencies to the gold standard, BR* again has larger RMSE and mean interval scores when estimating $\beta_2,$ the treatment arm indicator.
	
	An interesting case occurs when the dropout model is correctly specified but the imputation model omits $x_2$ and the $time \times x_2$ interaction.  First, consider the coefficient of the baseline indicator  $x_2,$ where all three estimators  AIPW-I, AIPW-S, and BR* perform as before, with minimal bias in accordance with their double-robustness.   Next, consider the coefficients of the omitted terms  $time$ and $time \times x_2.$ Here, AIPW-I and AIPW-S continue to obtain unbiased and efficient estimators for the coefficients of these terms.  This illustrates their double robustness, and that the imputation model does not need to accord with the estimation model. However in this setting  BR*,  when used as an imputation estimator, fails to consistently estimate the coefficients for $time$ and $time \times x_2$.  This is not entirely unexpected, because  the { BR*} approach is not developed for time varying predictors of interest, and also because it requires the variables  of interest in the final estimating equations to be included in the imputation models.  Thus, the BR* estimator is indeed not robust to this particular kind of model misspecification. 
	Paik's mean imputation derives an unbiased estimator for $\beta_2$, but fails to consistently estimate $\beta_3$ and $E(b_1)$. MMRM, WGEE and GEE-IND do not perform well in this scenario.  
	
	When both models are misspecified, none of the methods obtain consistent estimates for $\mathbf{\beta}$. As before, this indicates that the misspecification in the imputation and weighting models is substantial, and thus provides a good test of double robustness. Notably, in our simulation, AIPW-I and AIPW-S are more robust than other methods, especially for estimating $\mathbf{\hat \beta}$ for $x_2$ and $time$ by $x_2,$ the treatment and treatment by time interaction terms.  
	\subsection{Extreme Dropout Construct}
	We conduct a similar comparison of these estimators in the extreme dropout construct, where the mean dropout rates for $Y_2$ is 20\% and for $Y_3$ is 48\%.  Other simulation parameters are kept as previously described.  Results are qualitatively similar to the moderate dropout construct, with generally good performance of the three doubly-robust estimators compared to the  MMRM estimator when estimating $E[Y_3]$, and good performance of AIPW-I and AIPW-S when estimating $\mathbf{\beta}$ through the fully imputed data sets. In some scenarios, the loss of efficiency of AIPW-S relative to AIPW-I becomes apparent, although its performance is still acceptable. Compared with the moderate construct, WGEE has even more convergence issues reported.  Details are given in the Appendix \ref{secapp3}.

	\section{Application to the MCI trial}
	
	The MCI trial has been described in section \ref{MCI}.  The primary outcome of the trial was time to progression to Alzheimer's disease (AD).  The main conclusion of the trial was that Vitamin E had no benefit, while donepezil provided some benefit over placebo at 12 months, but not at 36 months, in accordance with the known symptomatic benefits of donepezil. The trial showed no benefit in secondary analyses comparing within-patient change on the two cognitive measures Mini-Mental State Examination (MMSE) and the Clinical Dementia Rating sum of boxes  (CDR-SOB) at 36 months.
	
	Here, for simplicity, only the donepezil arm and placebo arm are studied.  The two arms have similar demographic and clinical characteristics at baseline. The within-subject change from baseline for the MMSE and the CDR-SOB are used as our repeated measures outcomes. Lower MMSE (range: 0 to 30) and higher CDR-SOB (range: 0 to 18) indicate worse cognition. These measures are assessed at baseline, 12 months, 24 months, and 36 months. The mean values of MMSE for placebo group versus donepezil group are 27.3 (SD=1.80) versus 27.2 (SD=1.86) at baseline, 26.6 (SD=2.81) versus 27.1 (SD=2.46) at 12 months, 26.2 (SD=3.48) versus 26.4 (SD=3.14) at 24 months, and 25.0 (SD=5.05) versus 25.3 (SD=4.79) at 36 months. The mean values of CDR-SOB for placebo group versus donepezil group are 1.87 (SD=0.79) versus 1.78 (SD=0.80) at baseline, 2.28 (SD=1.45) versus 1.98 (SD=1.22) at 12 months, 2.72 (SD=2.02) versus 2.56 (SD=2.04) at 24 months, and 3.32 (SD=3.09) versus 3.20 (SD=2.66) at 36 months. The missing rates for the donepezil arm and the placebo arm are 26.1\% and 16.6\% at 12 months, 34.0\% and 29.3\% at 24 months, and 42.7\% and 32.0\% at 36 months. In order to ensure that the dropout is monotone, we use Paik's mean imputation \citep{ChoPaik1997} to fill in the intermediate missing responses. Only about 1.5\% of missing responses are intermediate missing values.
	
	We compare results of the two doubly-robust imputation estimators, AIPW-I and AIPW-S, and an estimator from a GEE model with an independence working correlation (GEE-IND). We also present unadjusted estimates of the mean at each time point, using all observed data. All models include the baseline outcome score, and model time as a linear trend, including both arm and the arm by time interaction as covariates, similar to one type of standard analysis model for AD trials.  Two primary estimands of interest are considered for each outcome. The first estimand is the mean difference between arms, placebo arm - treatment arm,  in within-patient change at the final visit (36 months), which is the primary estimand in many late stage AD trials. The least-squares mean difference, a model adjusted estimate, would be used \citep{lsmean}.  The second estimand of interest is the regression coefficient for the interaction between treatment arm and time. In addition, the model-adjusted estimates of the average difference between arms at intermediate time points (12 and 24 months), and the coefficient of time  are also presented, as these statistics would normally be computed in the final analysis of a clinical trial. The dropout model for the AIPW methods include arm, baseline outcome score and historical $Y$ as covariates.  Normal theory confidence intervals are constructed using the nonparametric bootstrap estimates of variance; the bootstrap sample size is 500.  
	
	Figure \ref{forestplot} displays the least-squares mean estimates for the mean difference between arms  in within-patient change for the MMSE (left panel) and the CDR-SOB (right panel), at 12, 24 and 36 months.  The thick transparent purple line shows the underlying mean differences at each time point in the observed data.  For the MMSE comparison, a positive difference is in favor of active treatment; for the CDR-SOB comparison, a negative difference is in favor of active treatment. Consistent with the known transient benefit of donepezil, the raw data for both measures show some benefit to the active arm at the 12 month assessment, which then diminishes towards zero at the 24 month and 36 month assessments. 
	
	The two doubly-robust imputation estimators (blue and green)  give similar results to each other,  showing a modest average benefit from doenpizil at the 12 months time point, as expected.  In all cases, this estimated early benefit tends to diminish over time, linearly because of the form of the model.The estimates made using the GEE-IND approach (orange line) diminish faster than the doubly-robust estimates, resulting in a smaller estimated treatment effect at the last visit. In fact, the GEE-IND estimate at 36 months concluded that the CDR-SOB change between the donepezil arm and the placebo arm was 0.001, which is in nominally favor of the placebo arm, while the AIPW-I and AIPW-S imputation methods obtained estimated values of -0.078 and -0.059 in support of the donepezil arm. Table \ref{table5} presents the estimates and standard errors. Since all the covariates are baseline characteristics, the AIPW-S method shows  good relative efficiency, with the smallest standard errors of 0.458 and 0.262 for MMSE and CDR-SOB at 36 months, respectively.
	
	
	\begin{figure}[h]   \centering
		\includegraphics[width=\textwidth]{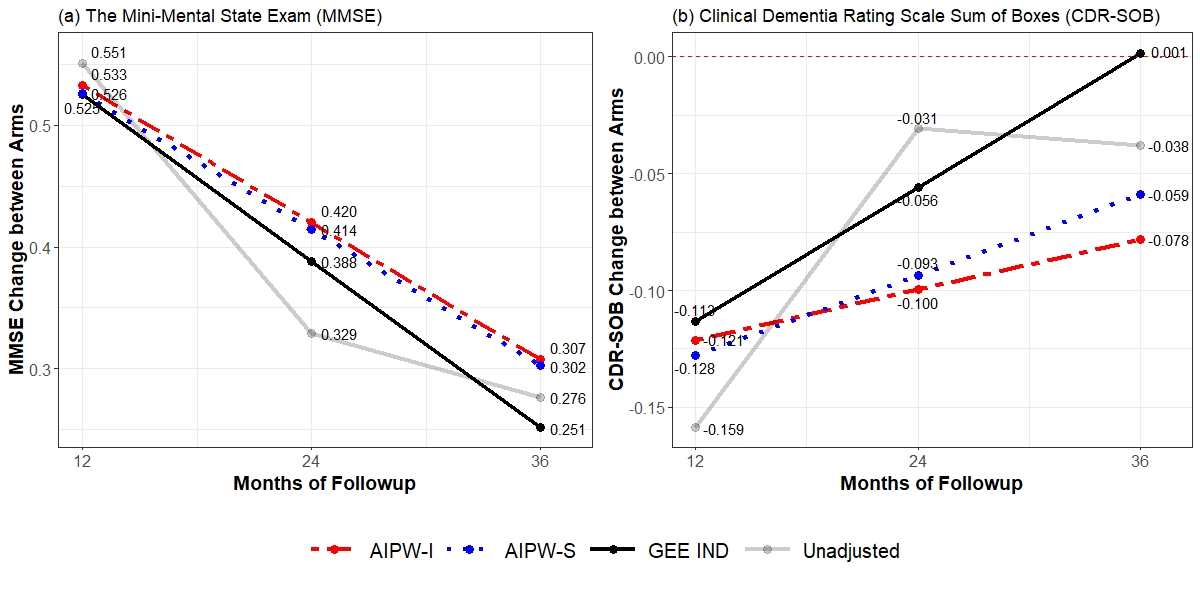}
		\caption{Least-Squares Estimates of Mean Difference between Donepezil and Placebo Arms in Within-Patient Change, for the MCI trial.  Left hand panel shows MMSE; a positive difference is in favor of active treatment.  Right hand panel shows CDR-SOB; a negative difference is in favor of active treatment. In both cases the apparent early benefit of donepezil (12 months time point) appears to diminish over time.  Gray solid line:  unadjusted mean difference, using the avialible data; Black solid line: adjusted estimates from a GEE model with independent working correlation; Blue dotted line: doubly-robust estimate using the simplified AIPW-S imputation for missing outcomes; Red dashed line: doubly-robust estimate using AIPW-I to adjust for missing outcomes.} \label{forestplot}
	\end{figure}
	
	\begin{table*}
		\caption{MCI trial, Estimates of Regression Coefficients for \textit{time} and \textit{time:arm}, and Least-Squares Estimates of Mean Difference between Arms in Within-Patient Change, at 12, 24 and 36 Months.
		}
		\label{table5} 
		\renewcommand{\arraystretch}{1.1}
		\centering
		\resizebox{\textwidth}{!}{%
			\begin{tabular}{lllllllllllllll}
				\hline
				\hline
				& \multicolumn{2}{c}{time} & & \multicolumn{2}{c}{time:arm} & & \multicolumn{2}{c}{Diff. in 1-yr change} & & \multicolumn{2}{c}{Diff. in 2-yr change} & & \multicolumn{2}{c}{Diff. in 3-yr change}\\
				\hline
				& Estimate & SE & & Estimate & SE & & Estimate & SE & & Estimate & SE & & Estimate & SE \\ 
				\hline
				\multicolumn{15}{l}{\textit{Mini-Mental State Exam (MMSE)}} \\ 
				GEE-IND &  -0.803 & 0.148 &  & -0.137 & 0.226 &  & 0.525 & 0.224 &  & 0.388 & 0.285 &  & 0.251 & 0.463 \\ 
				AIPW-I & -0.924 & 0.138 &  & -0.113 & 0.224 &  & 0.533 & 0.223 &  & 0.420 & 0.290 &  & 0.307 & 0.468 \\ 
				AIPW-S & -0.909 & 0.138 &  & -0.112 & 0.218 &  & 0.526 & 0.222 &  & 0.414 & 0.286 &  & 0.302 & 0.458 \\ 
				\hline
				\multicolumn{15}{l}{\textit{Clinical Dementia Rating Scale (CDR) Sum of Boxes}} \\ 
				GEE-IND   & 0.538 & 0.087 &  & 0.057 & 0.122 &  & -0.113 & 0.113 &  & -0.056 & 0.173 &  & 0.001 & 0.277 \\ 
				AIPW-I  & 0.632 & 0.081 &  & 0.022 & 0.111 &  & -0.121 & 0.110 &  & -0.100 & 0.168 &  & -0.078 & 0.263 \\ 
				AIPW-S  & 0.625 & 0.081 &  & 0.034 & 0.110 &  & -0.128 & 0.109 &  & -0.093 & 0.168 &  & -0.059 & 0.262 \\  
				\hline
				\hline
			\end{tabular}
		}
		\begin{tablenotes}
			\item Standard errors were derived from 500 times nonparametric bootstrap.
		\end{tablenotes}
	\end{table*}
	
	
	Table \ref{table5} also shows the estimated values of $\beta_{time}$ and $\beta_{time:arm}$ by the two AIPW imputation methods and the GEE-IND model. Again, the AIPW-I and AIPW-S obtained similar estimated values, while overall AIPW-S had smaller estimated standard errors. GEE-IND reported different estimated values and greater standard errors comparing with the two doubly-robust imputation methods. 
	
	
	\section{Discussion}
	
	In this paper, using the approach of \cite{Seaman2009}, we have developed an imputation framework for AIPW-based doubly-robust estimators which is suitable for a class of general longitudinal estimating equations when the data are observed with monotone dropout under MAR.  
	Confidence intervals are constructed using the bootstrap, and the estimators can be implemented using standard software tools.  We develop two specific imputation estimators, AIPW-I and AIPW-S, that appear to have reliable performance within the framework.  The simpler AIPW-S estimator uses only baseline values in the imputation models and reduces the number of models that need to be estimated from $M(M-1)/2$ to maximum $M-1$, at the cost of some stronger assumptions and potential loss of efficiency.   We show that the imputed completed data  from either AIPW-I or AIPW-S can  be substituted into a full data estimating equation $\sum_{i=1}^N U_i(\mathbf{\beta}) =0$, to obtain an estimator which inherits  doubly-robust properties from the imputation step.  In particular, as is the case in other imputation settings, a given APIW-imputed data set can be used across several different estimands, as the imputation models and the analysis models do not need to be coupled to one another.  For example, a single AIPW-imputed data set might be used across several primary or secondary analyses in a clinical trial. Under a pre-specified estimand, the AIPW-I estimator has the same workflow as the estimator of \cite{Seaman2009}, however, once  a completed data set has been created it can be used for variety of other estimands without re-fitting the doubly-robust model again. AIPW-S has considerably simplified computational demands compared to AIPW-I.

	We conducted extensive simulation studies across different estimands to evaluate the performance of these AIPW-based imputation estimators, including comparison with the approach proposed by Bang and Robins (BR*), which incorporates the probability weights inside the estimating equations as an additional covariate. The BR* estimator as originally proposed is suitable for estimating  the outcome at a final time point, $E[Y_{iM}]$.  Interestingly, we found only a few prior studies which make similar comparisons in the longitudinal setting \citep{Tsiatis2011, Seaman2009, Schnitzer2016}.   
	
	When the estimand of interest was  $E[Y_{iM}]$ or was the regression coefficient of a baseline covariate, the three doubly-robust estimators,  BR*, AIPW-I, and AIPW-S, were consistent as expected, even when the imputation models or the dropout models were badly misspecified. Furthermore, these doubly-robust estimators appeared to be competitive in efficiency with the correctly specified MLE estimators, when the imputation models were correctly specified.  Bootstrap confidence intervals had correct coverage probabilities, with reasonable bootstrap sample sizes.
	
	When the primary estimand  was  the coefficient of  a time-varying covariate, the two AIPW imputation methods, AIPW-I and AIPW-S, again demonstrated performance comparable to efficient MLE estimators whenever the imputation models were correctly specified.  When the imputation model was incorrect but the dropout model was correct, the two AIPW imputation estimators still performed well, with low bias and good coverage probabilities, although with some loss of efficiency.  In particular, we  studied an extreme case of a misspecified imputation model, which omitted the time-varying covariate of interest (i.e. the estimand in the analysis model was the coefficient of a  time-variable which was omitted from the imputation model).  This illustrates a case when the imputation model might be constructed prior to considering the estimand of interest, such as could be the case when using already imputed data to study a new estimand.   The two imputation-framework estimators worked well in this case.

	
	In this  last scenario, we also compared  use of the  BR* estimator for the imputation step within the imputation framework.  Our rationale for doing this is because, in the case when the imputation model is correct and agrees with the analysis model, the BR* estimator is similar in form to an imputation estimator.  However, when used as an imputation estimator in the extreme case described above, where the estimand in the analysis model was the coefficient of a  time-varying variable which was omitted from the imputation model, 
	the result had unacceptably large bias compared to AIPW-I or AIPW-S.  This is not unexpected, because the
	BR* estimator is as yet only demonstrated to work in the case of baseline covariates of interest and outcome at a given time point, and we included time-varying covariates and longitudinal outcomes. Additionally, our misspecified models omitted one of the variables in the final estimating equations,  a common approach to model misspecification in prior simulation studies \citep{Tsiatis2011, Seaman2009, Schnitzer2016}.   In our case, we omitted one of the variables that defines the target estimand, a case of potentially severe model misspecification.  Notably,  the BR* estimator requires that the $X$ which define the estimand of interest are  included in the imputation models.  This adds a requirement that the imputation models  be  tailored to the final estimand, a requirement that our imputation-based approach is  designed to avoid.  
	
	This AIPW-based doubly-robust imputation framework provides a good opportunity for sensitivity analysis in randomized trials and other settings. For example, in our application to a trial of donepezil for Alzheimer's disease, it was clear that differential dropout might introduce significant bias in the comparisons of interest between study arms (Figure \ref{MCIplot}). Because the dropout rates depended on disease severity and differed significantly between the donepezil and placebo arms, it is unclear whether a standard regression model using a GEE approach would produce consistent estimates. Our application of doubly-robust imputation to the donepezil data illustrates a case where the direction of the estimated treatment effect might be reversed by using a doubly-robust estimator, indicating significant bias in the standard MMRM estimates. Indeed, our simulation studies demonstrated how a mixed-effects model with a simplified covariance matrix, such as is not infrequently used in clinical trials, can produce badly biased results. In such a situation, doubly-robust AIPW-imputation based estimators such as AIPW-I and AIPW-S might be helpful as a sensitivity analysis for the primary analysis. The same data might then be useful for additional analyses, which is an advantage of the imputation-based approach.   Importantly, the theoretical basis for consistency  is well established for these DR imputation estimators in the setting of informative longitudinal dropout, which is not the case, to our knowledge, for the imputation by chained equations approach in common use.  Future work will compare these doubly-robust estimators to the more usual multiple imputation approaches used as sensitivity analyses in the clinical trials setting.


	
	
	\begin{acks}[Acknowledgments]
		The authors gratefully acknowledge Dr. Howard Feldman, PI of the Alzheimer's Disease Cooperative Study, for his mentorship and support of the first author during his Ph.D. studies at University of California San Diego. 
		
		The authors also thank editor Jeffrey Morris, the corresponding associate editor, and two referees for insightful comments which led to important improvements over the earlier drafts.
	\end{acks}
	
	\begin{funding}
		The authors gratefully acknowledge support from grant number R01AG061146 from the US DHHS NIH National Institute on Aging, and from grant number P30 CA023100 from the NIH National Cancer Institute. 
		
		Data collection and sharing for this project were obtained from the Alzheimer’s Disease Cooperative Study (ADCS), funded by the National Institutes of Health Grant U19 AG010483.
		
		
	\end{funding}

	\begin{appendix}
		\section{Double robustness of AIPW-I}\label{secApp1}
		Recall that  $\mathbf{\beta}$ is the parameter of interest, which is defined as the solution to  the full data estimating equations $E[U(\bar L_{M}, \mathbf{\beta}, \varphi_0)]=0,$ where here the $\varphi$ are full-data nuisance parameters, assumed for now to be known, and the estimating functions $U$ are taken to be orthogonal to the full data tangent space generated by the $\varphi$.  Following 
		\cite{Tsiatis2006}, we will sometimes suppress the explicit dependence of  $U$ on $\varphi$. Suppose there is a known monotone missingness mechanism $\lambda_j(\alpha_0) =Pr(R_{j} = 0 | \bar L_{j-1}, R_{j-1} = 1, \alpha_0)$.   In this case a semiparametric estimator for $\mathbf{\beta}$  using the observed data can be obtained by deriving elements orthogonal to the nuisance tangent spaces for $\varphi$ and $\alpha$, and using these to
		to adapt $U$ into the observed data estimating equations.  In our setting of monotone dropout, this approach leads to an  AIPW estimating equations of the form  
		\citep[equation 9.4 and Theorem 9.2]{Tsiatis2006}
		\begin{equation}\label{appendA1}
			E\bigg[   \frac{C_{M}}{\pi_{M}(\alpha)}U(\bar L_{M}, \mathbf{\beta}) +  \displaystyle\sum_{j=1}^{M-1} \bigg(\frac{C_{j} - \lambda_{j+1} (\alpha)R_{j}}{\pi_{j+1} (\alpha)}\bigg) h^{j}(\bar L_{j}, \mathbf{\beta}, \varphi) \bigg)\bigg] =0.
		\end{equation}
		for arbitrary $h_j$. Motivated by 
		\cite{Tsiatis2006} Theorem 10.1,  which gives the form of the optimal observed-data influence function, one may take the $h^j$ to have the corresponding form, that is $ h^{j}(\bar L_{ij}, \mathbf{\beta}, \varphi) = E[U(\bar L_{i,M}, \mathbf{\beta}, \varphi )| \bar L_{ij}, R_{ij} = 1]. $ Under suitable smoothness conditions on the functions $U$ and $h$, and moment conditions on $ L_{M}$, standard arguments for m-estimators show that the solution $\mathbf{\hat \beta}$ to (\ref{appendA1}) is a consistent asymptotically normal estimator of $\mathbf{\beta}$ using the observed data, with influence function  as given in Tsiatis equation (9.5).

		Note that the parameters $\pi$ and $\lambda$ depend on $\alpha_0$, which is assumed known.  When $\alpha$ is unknown, it may be substituted by any efficient estimator   $\hat \alpha$ (such as the MLE), and $\mathbf{\hat \beta}$ will remain a consistent asymptotically normal estimator of $\mathbf{\beta}$ (Tsiatis Theorem 9.1).  Also, if $\hat \varphi$ is a root-n consistent estimator for $\varphi_0$, then, under regularity conditions, substituting $\hat \varphi$ into estimating equations (\ref{appendA1}) will not perturb the asymptotic distribution of $\mathbf{\hat \beta},$ since $U$ has been taken to be orthogonal to the nuisance parameter space (see the discussion following (9.3) in Tsiatis).

		The double robustness property concerns consistency of $\mathbf{\hat \beta}$ when 
		$
		\hat \alpha \xrightarrow[]{p}  \alpha^*  
		$
		and
		$
		\hat \varphi  \xrightarrow[]{p} \varphi^*  
		$
		but either $\alpha^* \neq \alpha_0$ or $\varphi^*  \neq \varphi$.  Here we give a simple and direct proof of the double robustness of $\mathbf{\hat \beta}$ in the longitudinal case.
		
		\begin{theorem}
			Suppose that $\hat \alpha \xrightarrow[]{p}  \alpha^*  $
			and
			$\hat \varphi^* \xrightarrow[]{p} \varphi^*  $
			and either a) $\alpha^* = \alpha_0$ or  b) $\varphi^*  =\varphi_0$. Then the solution $\mathbf{\hat \beta}$ to  the observed data estimating equations of the form (\ref{appendA1}), with $\hat \alpha$ substituted for $\alpha$ and $\hat \varphi$ substituted for $\varphi,$ is consistent for $\mathbf{\beta}$.
		\end{theorem}
		
		\begin{proof}
			Under regularity conditions, the estimating equations given by (\ref{appendA1}) may be written as 
			
			\begin{equation} \label{appendA2}
				\frac{1}{N} \displaystyle\sum_{i=1}^{N} \bigg( \frac{C_{i,M}}{\pi_{i,M}(\bar L_{i,M},  \alpha^*)}U(\bar L_{i,M}, \mathbf{\beta}) + \displaystyle\sum_{j=1}^{M-1} \bigg(\frac{C_{ij} - \lambda_{j+1}(\bar L_{ij}, \alpha^*) R_{j}}{\pi_{j+1}(\bar L_{ij}, \alpha^*)}\bigg) h^{j}\big(\bar L_{ij},   \varphi^* \big) \bigg)  =0 
			\end{equation} 
			plus a remainder term which is $o_p(1)$.
			Let  $\mathbf{\hat \beta^*}$ be the solution to the estimating equations given by (\ref{appendA2}).  Then using a Taylor expansion argument, under regularity conditions on $U, H$ and $\bar L_{M}$, it can be shown that $ \mathbf{\hat \beta^*}= \mathbf{\hat \beta} + o_p(1)$, so it is enough to consider $\mathbf{\hat \beta^*}.$
			To prove that $\mathbf{\hat \beta^*}$ is   consistent, it is enough to show  that $\mathbf{\beta}$ is a solution to the observed data estimating equations
			\begin{equation}\label{appendA3} 
				E \bigg[ \frac{C_{M}}{\pi_{M}(\bar L_{M}, \alpha^*)}U(\bar L_{M}, \mathbf{\beta}) + \displaystyle\sum_{j=1}^{M-1} \bigg(\frac{C_{j} - \lambda_{j+1}(\bar L_{j}, \alpha^*) R_{j}}{\pi_{j+1}(\bar L_{j}, \alpha^*)}\bigg) h^{j}\big( \bar L_{j},  \varphi^* \big) \bigg] = 0
			\end{equation}
			if either a) or b) holds.
			
			Recall that
			\begin{eqnarray}\label{appendA4}
				\frac{C_{M}}{\pi_{M}(\bar L_{M},  \alpha)} + \displaystyle\sum_{j=1}^{M-1} \bigg(\frac{C_{j} - \lambda_{j+1}(\bar L_{j},  \alpha) R_{j}}{\pi_{j+1}(\bar L_{j}, \alpha)}\bigg) = 1
			\end{eqnarray}
			
			Thus, equation (\ref{appendA3}) can be written as
			\begin{align*} 
				E \bigg[ \frac{C_{M}}{\pi_{M}(\bar L_{M},\alpha^*)}U(\bar L_{M}, \mathbf{\beta}) + \displaystyle\sum_{j=1}^{M-1} \bigg(\frac{C_{j} - \lambda_{j+ 1}(\bar L_{j},\alpha^*) R_{j}}{\pi_{j+1}(\bar L_{j},\alpha^*)}\bigg) h^{j}\big( \bar L_{j},  \varphi^*\big) \bigg] & = \\
				E \bigg[ U(\bar L_{M}, \mathbf{\beta}) + (\frac{C_{M}}{\pi_{M}(\bar L_{M},\alpha^*)} - 1)U(\bar L_{M}, \mathbf{\beta}) + \displaystyle\sum_{j=1}^{M-1} \bigg(\frac{C_{j} - \lambda_{j+1}(\bar L_{j},\alpha^*) R_{j}}{\pi_{j+1}(\bar L_{j},\alpha^*)}\bigg) h^{j}\big( \bar L_{j},  \varphi^*\big) \bigg] & = \\
				E\bigg[U(\bar L_{M}, \mathbf{\beta})\bigg] + E\bigg[\displaystyle\sum_{j=1}^{M-1} \bigg(\frac{C_{j} - \lambda_{j+1}(\bar L_{j},\alpha^*) R_{j}}{\pi_{j+1}(\bar L_{j},\alpha^*)}\bigg) \bigg(h^{j}\big( \bar L_{j},  \varphi^*\big) - U(\bar L_{M}, \mathbf{\beta})\bigg) \bigg] & = 0
			\end{align*}
			
			Since $E\bigg[U(\bar L_{M}, \mathbf{\beta})\bigg] = 0$, it is enough to show that
			\begin{eqnarray}\label{appendA5}
				E_i\bigg[ \bigg(\frac{C_{j} - \lambda_{j+1}(\bar L_{j},\alpha^*) R_{j}}{\pi_{j+1}(\bar L_{j},\alpha^*)}\bigg) \bigg(h^{j}\big( \bar L_{j},  \varphi^*\big) - U(\bar L_{M}, \mathbf{\beta})\bigg) \bigg] & = 0
			\end{eqnarray}
			for all $j = 1,...,M-1$, if if either a) or b) holds.
			We may take iterated expectations to show that
			
			\begin{align*} 
				E \bigg[ \bigg(\frac{C_{j} - \lambda_{j+1}(\bar L_{j},\alpha^*) R_{j}}{\pi_{j+1}(\bar L_{j},\alpha^*)}\bigg) \bigg(h^{j}\big( \bar L_{j}, \hat \varphi \big) - U(\bar L_{M}, \mathbf{\beta})\bigg) \bigg] & = \\
				E \bigg\{ E\bigg[ \bigg(\frac{C_{j} - \lambda_{j+1}(\bar L_{j},\alpha^*) R_{j}}{\pi_{j+1}(\bar L_{j},\alpha^*)}\bigg) \bigg(h^{j}\big( \bar L_{j}, \hat \varphi \big) - U(\bar L_{M}, \mathbf{\beta})\bigg) \bigg| \bar L_{j}, R_{j} \bigg] \bigg\} & = \\
				E \bigg[ \bigg(\frac{E(C_{j} | L_{j}, R_{j}) - \lambda_{j+1}(\bar L_{j},\alpha^*) E(R_{j}| L_{j}, R_{j})}{\pi_{j+1}(\bar L_{j},\alpha^*)}\bigg) \bigg(h^{j}\big( \bar L_{j},  \varphi^* \big) - E[U(\bar L_{M}, \mathbf{\beta}, \varphi_0 )| \bar L_{j}, R_{j} ])\bigg) \bigg]
			\end{align*}
			Consider the numerator of the left hand factor above.  If $R_j=0$, both terms in this numerator are zero, so the difference is zero.  If $R_j=1$ and a) holds, then  $\lambda_{j+1}(\bar L_{j},\alpha^*) = Pr(R_{j+1} = 0 | L_{j}, R_{j} = 1) = Pr(C_{j} = 1 | L_{j}, R_{j} = 1)$, and again the difference is zero.  This demonstrates (\ref{appendA5}). If b) holds, that is, if  $ h^{j}(\bar L_{j}, \mathbf{\beta}, \varphi^*) = E[U(\bar L_{M}, \mathbf{\beta}, \varphi_0 )| \bar L_{j}, R_{j} = 1] $,  then the right hand factor is zero, establishing the remaining case.
		\end{proof}
		
		%
		
		
		\section{Derivation of AIPW-S and Demonstration of its Double Robustness}\label{secapp2}
		
		We will find a short lemma to be convenient.
		
		\begin{lemma} The following algebraic equalities hold:
			
			\begin{equation}\label{appendB0}
				\displaystyle\sum_{j=l}^{M} (\frac{C_j - \hat \lambda_{j+1} R_j}{\hat \pi_{j+1}}) = \frac{R_l}{\hat \pi_l}
			\end{equation}
			
			and
			
			\begin{equation}\label{appendB00}
				\displaystyle\sum_{t=1}^{l-1} (\frac{C_t - \hat \lambda_{t+1} R_t}{\hat \pi_{t+1}}) = 1 - \frac{R_l}{\hat \pi_l}. 
			\end{equation}
		\end{lemma}
		
		\begin{proof}
			Equation (\ref{appendB0}) follows from considering three cases.
			First, if $R_l = 0$, by the definition of monotone missing, $\forall k \geq l$, we have $R_k = 0$ and $C_k = 0$. Thus, $\displaystyle\sum_{j=l}^{M} (\frac{C_j - \lambda_{j+1} R_j}{\pi_{j+1}}) = \frac{R_l}{\pi_l} = 0$.\\
			Second, consider the case where $R_l = 1$ and there is a $  k$ with $ l < k \leq M$ such that $R_k = 0$ and $R_{k-1} = 1$. By the definition of monotone missing and hazard function, $C_{j} = 1 \text{ if } j = k-1$, otherwise $C_{j} = 0$. $R_j = 1 \text{ if } j<k$, otherwise $R_j = 0$. $\hat \pi_l = \prod^l_{j=1} (1 - \lambda_j)$. Thus,
			\begin{align*} 
				\displaystyle\sum_{j=l}^{M} (\frac{C_j - \hat \lambda_{j+1} R_j}{\hat \pi_{j+1}}) & = 
				\frac{-\hat \lambda_{l+1}}{\hat \pi_{l+1}} + \dots + \frac{-\hat \lambda_{k-1}}{\hat \pi_{k-1}} + \frac{1-\hat \lambda_{k}}{\hat \pi_{k}} + 0 + \dots + 0 \\ 
				& = \frac{R_l}{\hat \pi_l} = \frac{1}{\hat \pi_l}
			\end{align*}
			Third, if $R_{iM} = 1$, $\forall k \leq M$, $R_k = 1$. Thus, $R_l = 1$ as well. By the definition of monotone missing and hazard function, we have
			\begin{align*} 
				\displaystyle\sum_{j=l}^{M} (\frac{C_j - \hat \lambda_{j+1} R_j}{\hat \pi_{j+1}}) & = 
				\frac{-\hat \lambda_{l+1}}{\hat \pi_{l+1}} + \dots + \frac{-\hat \lambda_{M}}{\hat \pi_{M}} + \frac{1}{\hat \pi_{M}} \\ 
				& = \frac{R_l}{\hat \pi_l} = \frac{1}{\hat \pi_l}
			\end{align*}
			\\
			This demonstrates (\ref{appendB0}). Equality (\ref{appendB00}) then follows immediately from (\ref{appendB0}) and (\ref{appendA4}), which completes the proof.  
		\end{proof}
		

		As before, consider estimating equations (\ref{appendA1}).  We will suppress the dependence on $\alpha$ as this is not our focus here.  Define $\pi_{M+1}=\pi_{iM}$, $\lambda_{M+1}=0$, and $h^M(\mathbf{\beta}) = U(\bar L_{i,M}, \mathbf{\beta}).$ Then, the estimating equations for (\ref{appendA1}) can be written more simply as
		\begin{equation}\label{appendB1}
			\displaystyle\sum_{i=1}^{N} \bigg( \displaystyle\sum_{j=1}^{M} \bigg(\frac{C_{ij} - \lambda_{i,j+1} R_{ij}}{\pi_{i,j+1}}\bigg) h^{j}(\bar L_{ij}, \mathbf{\beta}, \varphi) \bigg) = 0
		\end{equation}
		This follows because, 1) if $R_{iM} = 0$, then $\frac{C_{M}}{\pi_{M}} = \frac{C_{M} - \lambda_{M+1} R_{M}}{\pi_{M+1}} = 0$; 2) if $R_{iM} = 1$, then $\frac{C_{M}}{\pi_{M}} = \frac{C_{M} - \lambda_{M+1} R_{M}}{\pi_{M+1}} = \frac{1}{\pi_{M}}$. \\ 
		
		Now consider an estimating functions $U$ of the form $U(\bar L_{M}, \mathbf{\beta}) = q(\mathbf{\beta})(\mathbf{Y} - \mathbf{\mu (\mathbf{\beta})})$, where $\mathbf{Y} = (Y_{1}, Y_{2},..., Y_{M})^T$ and $\mathbf{\upmu} = (\mu_{1}, \mu_{2},..., \mu_{M})^T$, and where $q(\mathbf{\beta}) = \frac{\partial \mathbf{\upmu}^T}{\partial \mathbf{\beta}} V^{-1}$. Furthermore, suppose $\hat q(\mathbf{\beta})$ is a consistent estimator of $q(\mathbf{\beta})$, such as is readily available for the canonical generalized estimating equations. Recall that $ h^{j}(\bar L_{j}, \mathbf{\beta}, \varphi) = E[U(\bar L_{M}, \mathbf{\beta}, \varphi )| \bar L_{j}, R_{j} = 1], $ and let $\hat h^j =  h^{j}(\bar L_{j}, \mathbf{\beta}, \hat \varphi)$ for $ \hat \varphi $ an estimate of $\varphi$.

		In this special case, $\hat h^j$ takes the simple form:
		\begin{align*} 
			\hat{h}^1 & =\hat q(\mathbf{\beta}) (
			\begin{bmatrix}
				y_1 & \hat y_2^{(1)} & \hat y_3^{(1)} & \cdots  & \hat y_{iM}^{(1)}
			\end{bmatrix}^T - \mathbf{\upmu} )\\ 
			\hat{h}^2 & =\hat q(\mathbf{\beta}) (
			\begin{bmatrix}
				y_1 & y_2 & \hat y_3^{(2)} & \cdots  & \hat y_{iM}^{(2)}
			\end{bmatrix}^T - \mathbf{\upmu} )\\ 
			\vdots \\ 
			\hat{h}^M & =\hat q(\mathbf{\beta}) (
			\begin{bmatrix}
				y_1 & y_2 & y_3 & \cdots  & y_{iM}
			\end{bmatrix}^T - \mathbf{\upmu} )
		\end{align*}
		where  $\hat y^j_k$ is the corresponding estimate of  $E[Y_k | \bar L_{j}, R_{j} = 1, \mathbf{\beta}, \hat \varphi] ,$  for $k >j$.
		Thus we can write:
		\begin{equation}\label{appendB2}
			\displaystyle\sum_{j=1}^{M} (\frac{C_j - \hat \lambda_{j+1} R_j}{\hat \pi_{j+1}}) \hat h^j =
			\hat q(\mathbf{\beta})
			\begin{bmatrix}
				\displaystyle\sum_{j=1}^{M} (\frac{C_j - \hat \lambda_{j+1} R_j}{\hat \pi_{j+1}}) (y_1-\mu_1) \\[0.4em]
				\vdots \\[0.4em]
				\displaystyle\sum_{j=l}^{M} (\frac{C_j - \hat \lambda_{j+1} R_j}{\hat \pi_{j+1}}) (y_l-\mu_l) +
				\displaystyle\sum_{t=1}^{l-1} (\frac{C_t - \hat \lambda_{t+1} R_t}{\hat \pi_{t+1}}) (\hat y_l^{(t)}-\mu_l) \\[0.4em]
				\vdots \\[0.4em]
				\frac{C_{iM}}{\hat \pi_{M}} (y_{iM}-\mu_{iM}) +
				\displaystyle\sum_{t=1}^{M-1} (\frac{C_t - \hat \lambda_{t+1} R_t}{\hat \pi_{t+1}}) (\hat y_{iM}^{(t)}-\mu_{iM})
			\end{bmatrix}
		\end{equation}
		\\

		Now consider the case where $\hat y_l^{(j)}$ is  an estimate of $E[Y_l |X_0, v]$ independent of j, where $X_0$ contains baseline covariates and $v$ indicates \textit{visit} or \textit{time}, such as a linear mixed-effects model. Then $\hat y_l^{(1)} = \hat y_l^{(2)} = ... = \hat y_l^{(M)}$. Equation (\ref{appendB2}) then can be simplified as:
		\begin{equation}\label{appendB3}
			\displaystyle\sum_{j=1}^{M} (\frac{C_j - \hat \lambda_{j+1} R_j}{\hat \pi_{j+1}}) \hat h^j =
			\hat q(\mathbf{\beta})
			\begin{bmatrix}
				\frac{ R_1}{\hat \pi_{1}} (y_1-\mu_1) \\[0.4em]
				\vdots \\[0.4em]
				\frac{R_l}{\hat \pi_{l}} (y_l-\mu_l) +
				(1 - \frac{R_l}{\hat \pi_{l}}) (\hat y_l - \mu_l) \\[0.4em]
				\vdots \\[0.4em]
				\frac{C_{iM}}{\hat \pi_{M}} (y_{iM}-\mu_{iM}) +
				(1 - \frac{C_{iM}}{\hat \pi_{M}}) (\hat y_{iM} - \mu_{iM})
			\end{bmatrix}
		\end{equation}
		
		
		\subsection{Proof of double robustness of AIPW-S\label{app2.1a}}
		Following equation (\ref{appendA1}), we can obtain similar estimating equations for AIPW-S:
		\begin{equation}\label{appendB10}
			\displaystyle\sum_{i=1}^{N} \bigg( \frac{C_{i,M}}{\pi_{i,M}}U(\bar L_{i,M}, \mathbf{\beta}) + \displaystyle\sum_{j=1}^{M-1} \bigg(\frac{C_{ij} - \lambda_{i,j+1} R_{ij}}{\pi_{i,j+1}}\bigg) h(X_{i,0}, v, \mathbf{\beta}) \bigg) = 0
		\end{equation}
		
		where $ h(X_{i,0}, v, \mathbf{\beta}) = E[U(\bar L_{i,M}, \mathbf{\beta}) | X_{i,0}, v] $.
		Thus, to show that AIPW-S is doubly-robust, we may use similar arguments as in \textit{Theorem A1}. In particular, we can write an equation for  AIPW-S that is similar to equation (\ref{appendA5}):
		\begin{eqnarray}\label{appendB11}
			E_i\bigg[ \bigg(\frac{C_{ij} - \lambda_{i,j+1}(\bar L_{ij},\hat \alpha) R_{ij}}{\pi_{i,j+1}(\bar L_{ij},\hat \alpha)}\bigg) \bigg(h\big( X_{i,0}, v, \hat \varphi \big) - U(\bar L_{i,M}, \mathbf{\beta})\bigg) \bigg] & = 0
		\end{eqnarray}
		for all $j = 1,...,M-1$.
		In the case where $X$ and $\nu$ are sufficient to render the data MAR, by definition of the law of iterated expectation, it is straightforward to show the double robustness for equation (\ref{appendB11}) following  Appendix A.
		\clearpage
		
		\section{Simulation methods specification}\label{apend3}
		\begin{table*}[h]
			\caption{Model specifications for the methods performed in simulation.
			}
			\renewcommand{\arraystretch}{1.1}
			\resizebox{\textwidth}{!}{
				\begin{tabular}{l|p{0.22\linewidth}|p{0.22\linewidth}|p{0.22\linewidth}|p{0.22\linewidth}}
					\hline
					\hline
					\\
					Methods & Y model correct & P model correct	& Y model incorrect	& P model incorrect \\ 
					\hline
					\multirow{2}{*}{BR*} & Recursive OLS & Logistic regression models & Recursive OLS & Logistic regression models \\\cline{2-5}
					                     & x1, x2, and historical Y & x2 and historical Y & x1, and historical Y & historical Y \\\cline{2-5}  
					\hline 
					\multirow{2}{*}{AIPW-I} & Paik's sequential OLS    &	Logistic regression models & Paik's sequential OLS    &	Logistic regression models \\\cline{2-5}
					& x1, x2, and historical Y  & x2 and historical Y & x1, and historical Y  & historical Y \\\cline{2-5}
					\hline
					\multirow{2}{*}{AIPW-S} & MMRM with unstructured correlation matrix   &	Logistic regression models & MMRM with unstructured correlation matrix   &	Logistic regression models\\\cline{2-5}
					& x1, x2, t and x2*t  & x2 and historical Y & x1 and t  & historical Y \\\cline{2-5}
					\hline
					\multirow{2}{*}{Paik} & Paik's sequential OLS    &	N/A & Paik's sequential OLS    & N/A \\\cline{2-5}
					& x1, x2, and historical Y  & N/A & x1, and historical Y  & N/A \\\cline{2-5}
					\hline
					\multirow{2}{*}{MMRM} & MMRM with unstructured correlation matrix    &	N/A & MMRM with unstructured correlation matrix    & N/A \\\cline{2-5}
					& x1, x2, t and x2*t  & N/A & x1 and t  & N/A \\\cline{2-5}
					\hline
					\multirow{2}{*}{WGEE} & Weighted GEE model with unstructured correlation matrix   &	Logistic regression models & Weighted GEE model with unstructured correlation matrix   &	Logistic regression models \\\cline{2-5}
					& x1, x2, t and x2*t  & x2 and historical Y & x1 and t  & historical Y \\\cline{2-5}
					\hline
					\multirow{2}{*}{GEE IND} & GEE model with independence correlation matrix   & N/A & GEE model with independence correlation matrix   &	N/A \\\cline{2-5}
					& x1, x2, t and x2*t  & N/A & x1 and t  & N/A \\\cline{2-5}
					\hline
					\hline
					
			\end{tabular}}
			\begin{tablenotes}
				\item 1. BR* represents Bang and Robins approach; AIPW-I represents AIPW imputation estimator; AIPW-S represents the simpler AIPW imputation estimator; Paik represents Paik's mean sequential imputation estimator; MMRM stands for mixed model with repeated measures; WGEE stands for weighted GEE; GEE IND stands for GEE with independence correlation matrix.
				\item 2. For each method, the first row indicates the model performed in the simulation, and the second row indicates the covariates controlled for in the model.
				\item 3. For BR* we further applied forward variable selection for fitting the Y models, as described in the main paper.
				\item 4. N/A refers to not available, which means the method does not fit in the given scenario.
			\end{tablenotes}
		\end{table*}
		
		\section{Simulation results in extreme dropout construct}\label{secapp3}
		We conduct a comparison of BR*, AIPW-I, AIPW-S, Paik's mean imputation, MMRM, WGEE, and GEE-IND in a extreme construct, where the mean dropout rates for $Y_2$ is 20\% and for $Y_3$ is 48\%. Generating model and dropout models have been defined in the main text section \ref{sec:Simulation}.
		\begin{table*}[h]
			\caption{Comparisons of $E(Y_3)$ among methods in six evaluations: 
				Bias,
				Root mean square error (RMSE), interval scores (Ints), coverage probability (Covp), Monte Carlo standard deviation (MCSD) and average standard errors (Ave SE), from 500 simulation runs and for $B=300$ bootstrap.
			}
			\label{table1 Extreme Scenario} 
			\renewcommand{\arraystretch}{1.1}
			\resizebox{\textwidth}{!}{
				\begin{tabular}{l|rrrrrrrrrrrrr}
					\hline
					\hline
					\\
					& Bias & RMSE & Ints & Covp & MCSD & Ave SE & & Bias & RMSE & Ints & Covp & MCSD & Ave SE \\ 
					\hline
					& \multicolumn{6}{c}{\underline{Y correct P correct}} & & \multicolumn{6}{c}{\underline{Y correct P incorrect}} \\
					BR*      & -0.02 & 0.32 & 1.45 & 0.94 & 0.32 & 0.31   &  & -0.02 & 0.31 & 1.43 & 0.95 & 0.31 & 0.31   \\
					AIPW-I  & -0.01 & 0.33 & 1.49  & 0.95 & 0.33 & 0.32   &  & -0.01 & 0.32 & 1.47  & 0.94 & 0.32 & 0.31   \\
					AIPW-S  & -0.02 & 0.34 & 1.51  & 0.94 & 0.34 & 0.32   &  & 0.04  & 0.33 & 1.50  & 0.94 & 0.33 & 0.32   \\
					Paik    & -0.01 & 0.31 & 1.42  & 0.95 & 0.31 & 0.31   &  & -0.01 & 0.31 & 1.42  & 0.95 & 0.31 & 0.31   \\
					MMRM    & -0.01 & 0.31 & 1.42  & 0.95 & 0.31 & 0.31   &  & -0.01 & 0.31 & 1.42  & 0.95 & 0.31 & 0.31   \\
					WGEE    & 0.00  & 0.32 & 1.55  & 0.95 & 0.32 & 0.35   &  & 0.03  & 0.32 & 1.53  & 0.94 & 0.32 & 0.34   \\
					GEE-IND & -0.14 & 0.33 & 1.57  & 0.92 & 0.30 & 0.30   &  & -0.14 & 0.33 & 1.57  & 0.92 & 0.30 & 0.30   \\
					\hline
					& \multicolumn{6}{c}{\underline{Y incorrect P correct}} & & \multicolumn{6}{c}{\underline{Y incorrect P incorrect}} \\
					BR*      & -0.22 & 0.42 & 1.98 & 0.92 & 0.36 & 0.40   &  &  -0.82 & 0.89 & 8.36 & 0.43 & 0.35 & 0.40   \\
					AIPW-I  & -0.02 & 0.35 & 1.54  & 0.95 & 0.35 & 0.34   &  & -0.76 & 0.84 & 8.71  & 0.41 & 0.35 & 0.36   \\
					AIPW-S  & -0.03 & 0.69 & 2.57  & 0.93 & 0.69 & 0.52   &  & -0.78 & 0.95 & 9.61  & 0.49 & 0.54 & 0.48   \\
					Paik    & -0.76 & 0.83 & 8.82  & 0.39 & 0.33 & 0.34   &  & -0.76 & 0.83 & 8.82  & 0.39 & 0.33 & 0.34   \\
					MMRM    & -0.67 & 0.75 & 6.66  & 0.49 & 0.32 & 0.34   &  & -0.67 & 0.75 & 6.66  & 0.49 & 0.32 & 0.34   \\
					WGEE    & 0.16  & 0.43 & 2.05  & 0.93 & 0.40 & 0.41   &  & -0.41 & 0.81 & 4.45  & 0.82 & 0.70 & 0.45   \\
					GEE-IND & -2.78 & 2.81 & 85.62 & 0.00 & 0.34 & 0.35   &  & -2.78 & 2.81 & 85.62 & 0.00 & 0.34 & 0.35   \\ 
					\hline
					\hline
					
			\end{tabular}}
			\\
			{\tiny

			}
		\end{table*}
		\clearpage
		\begin{table*}
			\caption{Comparisons of estimating regression coefficients among methods in four evaluations: 
				Bias,
				Root mean square error (RMSE), interval scores (Ints), coverage probability (Covp), from 500 simulation runs and for $B=300$ bootstrap.
			}
			\label{table2 Extreme Scenario} 
			\renewcommand{\arraystretch}{1.1}
			\centering
			\resizebox{\textwidth}{!}{%
				\begin{tabular}{l|cccccccccccccc}
					\hline
					\hline
					& \multicolumn{4}{c}{coefficient of $x_2$}                         &  & \multicolumn{4}{c}{coefficient of time}     &  & \multicolumn{4}{c}{coefficient of time:$x_2$}  \\ \cline{2-5} \cline{7-10} \cline{12-15} 
					& Bias                     & RMSE & Ints  & Covp &  & Bias  & RMSE & Ints   & Covp &  & Bias  & RMSE & Ints   & Covp \\ \hline
					& \multicolumn{14}{c}{Y correct P correct}                                                                           \\
					BR* & -0.01 & 0.22 & 2.04 & 0.98 &   & 0.00 & 0.11 & 0.53 & 0.93 &   & -0.01 & 0.15 & 0.69 & 0.95 \\ 
					AIPW-I & 0.01 & 0.10 & 0.50 & 0.95 &   & 0.01 & 0.11 & 0.54 & 0.94 &   & -0.01 & 0.15 & 0.70 & 0.94 \\ 
					AIPW-S & 0.01 & 0.10 & 0.50 & 0.94 &   & 0.01 & 0.11 & 0.54 & 0.94 &   & -0.01 & 0.15 & 0.71 & 0.94 \\ 
					Paik & 0.01 & 0.10 & 0.50 & 0.94 &   & 0.01 & 0.11 & 0.51 & 0.95 &   & -0.01 & 0.14 & 0.67 & 0.94 \\ 
					MMRM & 0.01 & 0.10 & 0.50 & 0.94 &   & 0.00 & 0.11 & 0.51 & 0.95 &   & -0.01 & 0.14 & 0.67 & 0.94 \\ 
					WGEE & 0.01 & 0.10 & 0.54 & 0.95 &   & -0.01 & 0.11 & 0.57 & 0.95 &   & 0.00 & 0.15 & 0.70 & 0.95 \\ 
					GEE-IND & 0.01 & 0.10 & 0.50 & 0.94 &   & -0.04 & 0.11 & 0.53 & 0.94 &   & 0.01 & 0.15 & 0.65 & 0.95 \\  \hline
					& \multicolumn{14}{c}{Y correct P incorrect}                                                                         \\
					BR* & -0.03 & 0.90 & 2.79 & 0.98 &   & 0.01 & 0.11 & 0.51 & 0.94 &   & -0.01 & 0.14 & 0.68 & 0.94 \\ 
					AIPW-I & 0.01 & 0.10 & 0.50 & 0.95 &   & 0.01 & 0.11 & 0.53 & 0.95 &   & -0.01 & 0.15 & 0.69 & 0.94 \\ 
					AIPW-S & 0.01 & 0.10 & 0.50 & 0.94 &   & -0.01 & 0.11 & 0.52 & 0.95 &   & 0.02 & 0.15 & 0.69 & 0.94 \\ 
					Paik & 0.01 & 0.10 & 0.50 & 0.94 &   & 0.01 & 0.11 & 0.51 & 0.95 &   & -0.01 & 0.14 & 0.67 & 0.94 \\ 
					MMRM & 0.01 & 0.10 & 0.50 & 0.94 &   & 0.00 & 0.11 & 0.51 & 0.95 &   & -0.01 & 0.14 & 0.67 & 0.94 \\ 
					WGEE & 0.01 & 0.10 & 0.54 & 0.95 &   & -0.01 & 0.11 & 0.54 & 0.95 &   & 0.02 & 0.15 & 0.74 & 0.94 \\ 
					GEE-IND & 0.01 & 0.10 & 0.50 & 0.94 &   & -0.04 & 0.11 & 0.53 & 0.94 &   & 0.01 & 0.15 & 0.65 & 0.95 \\  \hline
					& \multicolumn{14}{c}{Y incorrect P correct}                                                                         \\
					BR* & 0.01 & 0.12 & 2.60 & 0.98 &   & -0.90 & 1.00 & 22.89 & 0.04 &   & 1.76 & 1.94 & 65.81 & 0.00 \\ 
					AIPW-I & 0.01 & 0.10 & 0.52 & 0.95 &   & 0.00 & 0.13 & 0.63 & 0.94 &   & 0.00 & 0.16 & 0.75 & 0.95 \\ 
					AIPW-S & 0.00 & 0.11 & 0.54 & 0.95 &   & 0.00 & 0.13 & 0.64 & 0.94 &   & 0.00 & 0.17 & 0.79 & 0.95 \\ 
					Paik & 0.01 & 0.11 & 0.64 & 0.95 &   & -0.90 & 0.91 & 24.02 & 0.00 &   & 1.24 & 1.26 & 37.13 & 0.00 \\ 
					MMRM & 0.25 & 0.25 & 10.00 & 0.00 &   & -3.25 & 3.26 & 117.52 & 0.00 &   & 6.00 & 6.00 & 240 & 0.00 \\ 
					WGEE & 0.25 & 0.25 & 9.96 & 0.01 &   & -3.14 & 3.14 & 112.74 & 0.00 &   & 6.00 & 6.00 & 239.91 & 0.00 \\ 
					GEE-IND & 0.25 & 0.25 & 10.00 & 0.00 &   & -3.36 & 3.36 & 121.69 & 0.00 &   & 6.00 & 6.00 & 240.00 & 0.00 \\  \hline
					& \multicolumn{14}{c}{Y incorrect P incorrect}                                                                       \\
					BR* & -0.09 & 1.83 & 7.69 & 0.97 &   & -3.33 & 3.33 & 119.92 & 0.00 &   & 6.08 & 6.08 & 238.35 & 0.00 \\ 
					AIPW-I & 0.01 & 0.11 & 0.51 & 0.95 &   & -0.32 & 0.35 & 3.30 & 0.42 &   & 0.06 & 0.17 & 0.75 & 0.95 \\ 
					AIPW-S & 0.01 & 0.11 & 0.52 & 0.95 &   & -0.34 & 0.37 & 3.72 & 0.39 &   & 0.10 & 0.20 & 0.83 & 0.92 \\ 
					Paik & 0.01 & 0.11 & 0.64 & 0.95 &   & -0.90 & 0.91 & 24.02 & 0.00 &   & 1.24 & 1.26 & 37.13 & 0.00 \\ 
					MMRM & 0.25 & 0.25 & 10.00 & 0.00 &   & -3.25 & 3.26 & 117.52 & 0.00 &   & 6.00 & 6.00 & 240.00 & 0.00 \\  
					WGEE & 0.25 & 0.25 & 9.92 & 0.01 &   & -3.30 & 3.30 & 118.40 & 0.01 &   & 6.00 & 6.00 & 238.71 & 0.00 \\ 
					GEE-IND & 0.25 & 0.25 & 10.00 & 0.00 &   & -3.36 & 3.36 & 121.69 & 0.00. &   & 6.00 & 6.00 & 240.00 & 0.00 \\ 
					\hline 
					\hline
				\end{tabular}%
			}
		\end{table*}
	\end{appendix}

	\newpage
	\bibliographystyle{imsart-nameyear}
	\bibliography{AOAS_DR}
	
\end{document}